\documentclass[onecolumn,10pt]{IEEEtran}

\usepackage{times}
\usepackage{amssymb} %
\usepackage{amsmath}
\usepackage{amsthm} 
\usepackage{setspace}
\usepackage{booktabs} 

\usepackage{graphicx}
\usepackage{amsmath}
\usepackage{amsthm}
\usepackage{amsfonts}
\usepackage{color}
\theoremstyle{}
\newtheorem{theorem}{Theorem}
\newtheorem{lemma}{Lemma}

\newtheorem{remark}{Remark}

\usepackage{bm}
\usepackage{booktabs}
\usepackage{multirow}
\usepackage{algorithm}
\usepackage{algpseudocode}
\usepackage{subfigure}
\usepackage{amssymb}
\usepackage{cite}
\newcommand{\tabcaption}{\def\@captype{table}\caption}
\newcommand{\tabincell}[2]{\begin{tabular}{@{}#1@{}}#2\end{tabular}} 
\title{Explicit Constructions for Rack-Aware Minimum Storage Partially Cooperative Regenerating Codes}
\date{}
\author{\IEEEauthorblockN{Hengming Zhao, Dianhua Wu,  Minquan Cheng}
	\thanks{H. Zhao is with  School of Computer Science
		and Engineering, Guangxi Normal University
		Guilin, 541004, China, and also with School of Mathematics and Statistics, Nanning Normal University, Nanning 530100, China (e-mail: hengmingzh@163.com).}
	\thanks{M. Cheng is with the Key Lab of Education Blockchain and Intelligent Technology, Ministry of Education, and also with the Guangxi Key Lab of Multi-source Information Mining $\&$ Security, Guangxi Normal University, 541004 Guilin, China  (e-mail: chengqinshi@hotmail.com).}
	\thanks{D. Wu is with the School of Mathematics and Statistics, Guangxi Normal University, Guilin 541006, China (e-mail:  dhwu@gxnu.edu.cn).}
}
\begin{document}
		\maketitle
\begin{abstract}
The rack-aware storage model improves repair efficiency by exploiting locality within racks to minimize cross-rack traffic in a distributed storage system. While	the partially cooperative repair model presents a solution for multiple node failures that reduces the need to exchange data with all other host racks (defined as racks containing failed nodes), thus enhancing system flexibility. In this paper,  we focus on rack-aware  minimum storage partially cooperative regenerating (MSPCR) codes for repairing multiple node failures. We first derive the lower bound on the repair bandwidth for rack-aware MSPCR codes using extremal combinatorics, and then explicitly construct the first class of (asymptotically) optimal repair schemes for rack-aware MSPCR codes with a sub-packetization level of $(\bar{s}+\bar{h}-\delta)\bar{s}^{\bar{n}}$, which is smaller than that of the known rack-aware minimum-storage cooperative regenerating (MSCR)  codes when $\delta \geq 2$. By utilizing the grouping technique, we explicitly construct the second class of (asymptotically) optimal repair schemes for rack-aware MSPCR codes with a sub-packetization level  of $2^{\bar{n}}$. In particular, when $\delta=1$, our second codes reduce to rack-aware MSCR codes, while achieving an $(\bar{h}+1)$-fold reduction in sub-packetization level compared to the known rack-aware MSCR codes.
\end{abstract}

\section{Introduction}
Distributed  storage systems (DSSs) are extensively employed by various companies, such as Facebook and Google, leveraging multiple independent and unreliable devices (referred to as nodes) for data storage. In an $(n,k,d)$ distributed storage system, the original file is distributed across $n$ nodes, guaranteeing  that the original file can be reconstructed from any $k$ out of $n$ nodes. This characteristic is known as the $(n,k)$ maximum distance separable (MDS) property, where the parameter $k$ is referred to as the  reconstructive degree. When a node fails, it can be exactly repaired using $d$  helper nodes from the remaining $n-1$ survivor nodes. The storage capacity refers to the maximum data each node can store, whereas the repair bandwidth represents the total data transmitted by helper nodes to repair a failed node. These metrics reflect storage redundancy and repair efficiency, respectively. The first optimal trade-off between storage capacity and repair bandwidth was derived by Dimakis et al., in \cite{DGWWR} by cut-set bound where the minimum-storage regenerating (MSR) code and  minimum-bandwidth regenerating (MBR) code were respectively introduced at the two extreme points of this trade-off curve. For details of the MBR codes, please see \cite{Rashmi2011, Ernvall2014, Mahdaviani2019,Zhang2020,Wang2023} and the references therein.

The MSR codes have been widely studied due to the MDS property with the optimal repair bandwidth, see \cite{GFV,LLT,LTT,LWHY,Rashmi2011,RSE,TYLH,TWB,VBV,YB1,YB2,HLSH,C-Barg,HLH,ZZ,WC} and references therein. Of course, the MSR code can be regarded as an MDS array code with the optimal repair bandwidth. However, in numerous practical scenarios, especially in large-scale storage systems, multiple failures occur more frequently than a single failure. Furthermore, many systems (e.g., \cite{Cadambe2013}) adopt a lazy repair strategy. This strategy aims to restrict the repair cost of erasure codes. Instead of promptly repairing each individual failure, a lazy repair strategy defers action until $h$ failures take place, where $h\in[1,n-k]$. Subsequently, the repair  is carried out by downloading at least $h/(d-k+h)$ of the data amount from each helper node to recover the failed nodes. For the case of $h$ failures, all the studies can be divided into two repair models, i.e., the centralized repair which was first introduced in \cite{Cadambe2013} and the cooperative repair first introduced in \cite{HXWZL}. In the centralized model, the data center executes repairs for all failed nodes. Under the cooperative repair model, when repairing $h$ node failures, $h$ new nodes participate in a two-phase recovery process: first downloading data from $d$ helper nodes, then performing mutual data exchange to finalize node regeneration. This cooperative repair model offers distinct advantages in efficiency, reliability, scalability and latency, and has been extensively studied in \cite{SH,Li2014,LCT,YB,Ye,ZZW,Zhang2025} and references therein. 

The cooperative repair model faces two significant challenges. Firstly, transient node inaccessibility caused by network congestion may terminate repair procedures prematurely \cite{LCT}. Secondly, the inherent geographical distribution of new nodes across multiple clusters creates operational bottlenecks, as inter-cluster data transmission demands substantially greater bandwidth allocation than intra-cluster exchanges \cite{LSL}. To overcome these limitations, the authors in \cite{LSL} introduced the partially cooperative repair model, i.e., each new node is allowed to exchange data with some part of the other new nodes. In this case, some transiently inaccessible new nodes can be be ignored in the cooperative repair phase and the quantity of connections among different clusters could be decreased. They also derived a lower bound on the repair bandwidth for the partially cooperative repair model. The authors in \cite{LSL,LO,ZLC} proposed two constructions of MDS array codes for partially cooperative repair model. Unfortunately, neither of these two constructions achieves the lower bound on the repair bandwidth. Recently, the first explicit constructions of MSR codes for the partially cooperative repair model have been presented in \cite{LWCT}.

In a homogeneous distributed storage system, all nodes and the communication among them are regarded equally. Nevertheless, modern data centers frequently feature hierarchical topologies, which are formed by organizing nodes in racks, i.e., the rack-aware storage model. To provide enough bandwidth for the ever-larger racks, it is necessary to increasingly make use of specialized non-commodity switches and routers \cite{VFF}. In rack-aware storage systems, intra-rack communication is assumed to incur zero bandwidth cost during repair processes, then the repair bandwidth refers to the total amount of data transmitted in inter-rack communication for repairing failures.
In this case,  MDS array codes that achieve the lower bound on the repair bandwidth are referred to as rack-aware MSR codes.

When all failed nodes are in the same rack (this case can be roughly viewed as the centralized model), Zhang and Zhou \cite{Zhou2022} considered a relaxed repair model and enhanced erasure tolerance at the expense of storage efficiency, and they proposed minimum storage codes that sacrifice the MDS property. Later,
the authors in \cite{WZLT} derived a lower bound on the repair bandwidth and presented explicit constructions of rack-aware MSR codes. 
For the (fully) cooperative repair model, the authors of \cite{GDL} derived a lower bound on the repair bandwidth and proposed an explicit construction of rack-aware minimum-storage cooperative regenerating (MSCR) codes for the case where all host racks (racks with failed nodes) have an equal number of failures, which is less than some integer $u-v$.

\subsection{Research motivation and Contribution}
In this paper, we focus on the rack-aware MSR codes in the partially cooperative model where each host rack has the same number of failed nodes. When the number of failed nodes is inconsistent, we can treat this situation as the worst-case scenario where the number of failed nodes in each host rack is maximized.
In the  cooperative repair model, all remaining host racks participate in repairing failed nodes during the cooperative repair process, whereas the partially cooperative repair model utilizes only a subset of remaining host racks. This implies that for each host rack, we should carefully arrange which subset of remaining host racks should transmit specific data to it during the cooperative repair phase. Therefore,  the repair scheme proposed in \cite{GDL} can't be extended to the case of the partially cooperative repair model. As far as we know, there are few studies on the rack-aware MSR codes in the partially cooperative model. We find that by extending the construction method in \cite[Construction 2]{YB1} based on Hadamard designs and applying the grouping idea in the cooperative repair process, we can propose novel arrangements for subsets of remaining host racks to obtain two classes of rack-aware minimum storage partially cooperative regenerating (MSPCR) codes. Specifically the main results can be summarized as follows.  
\begin{itemize}
	\item We derive the lower bound on repair bandwidth for rack-aware MSPCR codes by using extremal combinatorics.
	\item According to this bound,  we obtain our first class of rack-aware MSPCR codes, i.e., the codes in Theorem \ref{th2}, by constructing a class of MDS array codes and stacking up multiple instances of the MDS array code. These codes achieve the minimum (optimal) repair bandwidth  when $b\in[1,u-v]$ and the asymptotically optimal repair bandwidth when  $b\in[u-v+1,u]$.
	\item When $\bar{d}=\bar{k}+1$, we obtain the second class of rack-aware MSPCR codes, i.e., the codes in Theorem \ref{th3}. By dividing each column of the code $\mathcal{C}$ into appropriate groups in the repair process, the subpacketization level is smaller by a factor of $\bar{h}-\delta+2$ than that in \cite{GDL} when $\delta=1$. These codes achieve the optimal repair bandwidth for $b\in[1,u-v]$ and the asymptotically optimal repair bandwidth for $b\in[u-v+1,u]$. 
\end{itemize}

A comparison between the main existing rack-aware MSCR codes and the new codes introduced in this paper is presented in Table \ref{tab1}.
\begin{table*}[http!]
	\renewcommand{\arraystretch}{1.2}
	\setlength\tabcolsep{3pt} 
	\centering
	\caption{The Existing Codes And New Codes In This Paper Where  $\bar{s}=\bar{d}-\bar{k}+1$, $k=\bar{k}u+v$, $b=h/\bar{h}$, And $k$, $\bar{n}$, $h$,  $\bar{h}$,  $\bar{d}$, $u$, $\delta$ Represent Code Dimension, The Number Of Racks, The Number Of  Failed Nodes, The Number Of Host Racks, The Number Of Helper Racks, The Size Of A Rack, And Cooperative Parameter, Respectively.
		\label{tab1}
	} 
	\begin{tabular}{|c|c|c|c|c|c|c|c|c|}	
		\hline
		&  sub-packetization  $l$ & $\delta$ &optimal repair bandwidth &  $\bar{d}$ & $b$ \\ 
		\hline
		\cite{GDL} &	 $(\bar{d}-\bar{k}+\bar{h})\bar{s}^{\bar{n}}$ & $1$ &  optimal& $\bar{k}+1\leq\bar{d}\leq \bar{n}-\bar{h}$  & $\leq u-v$\\ 		
		\hline
		Theorem \ref{th2} &	$(\bar{d}-\bar{k}+\bar{h}-\delta+1)\bar{s}^{\bar{n}}$  & $1\leq\delta\leq\bar{h}-2$&  optimal  & $\bar{k}+1\leq\bar{d}\leq \bar{n}-\bar{h}$  & $\leq u-v$\\
		\hline
		Theorem \ref{th2} &	$(\bar{d}-\bar{k}+\bar{h}-\delta+1)\bar{s}^{\bar{n}}$  & $1\leq\delta\leq\bar{h}-2$&  asymptotically  & $\bar{k}+2\leq\bar{d}\leq \bar{n}-\bar{h}$  & $u-v+1 \leq b\leq u$\\
		\hline
		Theorem \ref{th3} &	$2^{\bar{n}}$  & $1\leq\delta\leq\bar{h}-2$&  optimal  & $\bar{d}=\bar{k}+1$  & $\leq u-v$\\
		\hline
		Theorem \ref{th3} &	$2^{\bar{n}}$  & $1\leq\delta\leq\bar{h}-2$&  asymptotically  & $\bar{d}=\bar{k}+2$  & $u-v+1 \leq b\leq u$\\
		\hline
	\end{tabular}
\end{table*}

\subsection{Organization and notations} 

The rest of this paper is arranged as follows. In Section II, we describe the system model for multiple failures.  The main results are presented in Section III. The construction for the first class of rack-aware MSPCR codes, and the proof of Theorem \ref{th2} are shown in Section IV. Section V contains the construction of the
second class of rack-aware MSPCR codes and the proof of Theorem \ref{th3}. Finally, the concluding remarks are provided in Section VI.

In this paper, we use boldface capital letters, boldface lowercase letters, and calligraphic letters to denote matrices, vectors, and sets, respectively.
In addition, the following notations are used unless otherwise stated.
\begin{itemize}
	\item  $\mathbb{F}$ is a finite field.
	\item For an integer $n>0$, $[n]$ denotes the set $\{0,1,\ldots,n-1\}$.
	\item For any integers $a$ and $b$ with $a<b$, $[a,b)$ denotes the  set $\{a,a+1,\ldots,b-1\}$ , and
	$[a,b]=\left\{ a,a+1,\ldots,b\right\}$. 
	\item $|\cdot|$ denotes the cardinality of a set, $\oplus$ denotes the addition operation modulo $\bar{s}$, and $\top$ denotes the transpose operator.
	\item  Given an integer $a\in[s^n]$ where $n\in \mathbb{N}^+$ and $s\geq2$ is an integer, with $a=\sum_{i=0}^{n-1}a_is^i$ for integers $a_i\in[s]$, we refer to 
	${\bf a}=(a_{n-1},\ldots,a_1,a_0)$ as the $s$-ary representation of $a$, and let ${\bf a}(i,j)$ be the vector where the $i$-th entry of ${\bf a}$ is $j$ and the other entries are the same as that of ${\bf a}$, i.e., ${\bf a}(i,j)=(a_{n-1},\ldots,a_{i+1}, j,a_{i-1},\ldots,a_0)$.  We do not specifically distinguish $a$ and ${\mathbf{a}}$ in $[s^n]$.
	
	
\end{itemize}

\section{The System Model} 
\label{sect-system}
In this section, we will introduce rack-aware MSR codes under the partially cooperative repair model where the objective is to derive a new lower bound on the repair bandwidth, and explicitly construct two classes of rack-aware MSPCR codes.

In an $(n,k)$ distributed storage system, there is an original data ${\bf w}$ which is independently and identically
uniformly distributed in $\mathbb{F}^{kl}$. We divide the original data into $k$ blocks, denoted by ${\bf w}=({\bf w}_0^\top, {\bf w}_1^\top, \ldots, {\bf w}_{k-1}^\top)^\top$, where ${\bf w}_i=(w_{i,0},w_{i,1},\ldots,w_{i,l-1})^{\top}$, $i\in[k]$ is a column vector of size $l$. Then, for each $j\in[n]$, we use a generating submatrix $\mathbf{G}_{j}=(\mathbf{G}_{j,i})_{i\in[k]}$ where $\mathbf{G}_{j, i}$ is an $l\times l$ matrix over $\mathbb{F}$, encode ${\bf w}$ to obtain the coded data ${\bf c}_j$, i.e., $\mathbf{G}_{j}{\bf w}={\bf c}_j$, and store it in node $j$. Let the code  $\mathcal{C}=({\bf c}_0, {\bf c}_1, \ldots, {\bf c}_{n-1})$ where ${\bf c}_j=(c_{j,0},c_{j,1},\ldots,c_{j,l-1})^{\top}$, $j\in[n]$. We call $\mathcal{C}$ an $(n,k,l)$ maximum distance separable (MDS) array code if an original file can be recovered from any $k$ out of $n$ nodes. Here, $k$ can be seen as the code dimension, and the length of nodes $l$ is called  the sub-packetization level of the code in this paper. Given two positive integers $\bar{n}$ and $u$ such that $n=\bar{n}u$, we divide the nodes $\{0,1,\ldots,n-1\}$ into consecutive $\bar{n}$ groups (referred to as racks) each of which has $u$ nodes. Accordingly, the coordinates of code $\mathcal{C}$ are partitioned into segments of length $u$. We make no distinction between nodes and code coordinates, and refer to both of them as nodes. Each rack is equipped with a relayer node that has access to the contents of the other nodes in the same rack. The scenario where $u=1$ corresponds to the homogeneous model. In this paper, we focus on studying the case of $1<u\leq k$. Otherwise, a single node failure could be trivially repaired by the  surviving nodes within the same rack. Moreover, we also assume that $u\leq n-k$  to guarantee that the code has the ability to be repaired if full-rack fails\cite{WZLT}. Let $k=\bar{k}u+v$ for integers $\bar{k}\in[1,\bar{n})$ and $v\in [u]$ . Then we have $r=n-k=\bar{n}u-(\bar{k}u+v)=(\bar{n}-\bar{k})u-v=\bar{r}u-v$ where $\bar{r}=\bar{n}-\bar{k}$. Let $\bar{h}$ be the number of host racks, then $\bar{h}\in[1,\bar{r}]$.  The authors in \cite{GDL} studied the case that there are $\bar{h}$ host racks each of which has $b$ failed nodes and $b\in[1,u-v]$, where the number of failed nodes $h=b\bar{h}$. In practice, some host racks usually contain multiple failed nodes. If the number of failed nodes is inconsistent, we can treat this situation as the worst-case scenario where the number of failed nodes in each host rack is maximized. So,  we consider the case that each host rack contains $b={h}/{\bar{h}}$ failed nodes with $b\in[1, u]$. Let ${\cal F}=\{i_0,i_1,\ldots,i_{\bar{h}-1}\}$ be the index set of the $\bar{h}$ host racks. The authors in \cite{LO,LSL} introduced the partially cooperative repair model which can be regarded as an extension of the cooperative repair model. In this paper, we also use the partially cooperative repair model mentioned in \cite{LO,LSL}, that is, the rack-aware partially cooperative repair model consists of the following two phases: 
\begin{itemize}
	\item {\bf Download phase:} Let ${\cal R}=\{j_0,j_1,\ldots,j_{\bar{d}-1}\}\subseteq[\bar{n}]\setminus{\cal F}$ be the index set of the $\bar{d}\in[\bar{k},\bar{n}-\bar{h}]$ helper racks. Each rack $i_p\in {\cal F}$ $(0\leq p\leq \bar{h}-1)$ downloads $\beta_1(i_p,j)$
	coded symbols from each helper rack $j$ where $j\in {\cal R}$.
	\item {\bf Cooperative phase:} Each host  rack $i_p\in{\cal F}$  downloads $\beta_2(i_p,i')$ coded symbols from each host rack $i'\in{\cal S}_p\subseteq\mathcal{F}\setminus\{i_p\}$. We assume that $|\mathcal{S}_{0}|=|\mathcal{S}_{1}|=\dots=|\mathcal{S}_{{\bar{h}-1}}|=\bar{h}-\delta$ where $\delta\in[1,\bar{h}-1]$ is called cooperative parameter, i.e., each host rack is connected to the same number $\bar{h}-\delta$ of helper host racks. 	So, each host rack uses the symbols generated from the previous two phases together with the symbols from the $u-b$ surviving nodes in this host rack to recover the $b$ failed nodes.
\end{itemize}

Clearly, the number of helper nodes 
$d=\bar{d}u+u-b$. When $\delta=1$ or $\bar{h}=2$, the rack-aware partially cooperative repair model reduces into the rack-aware (fully) cooperative repair model that first introduced in \cite{GDL}. Moreover, if $u = 1$, this model is the cooperative repair model explored in \cite{HXWZL} and \cite{SH}.

Figure 1 uses a $(12,6)$ MDS array code to illustrate the two repair models: the rack-aware fully cooperative model and the rack-aware partially cooperative model. The former requires communication with all other $\bar{h}-1$ host racks, whereas the latter only needs to communicate with $\bar{h}-\delta$ host racks, where $\delta\in[1,\bar{h}-1]$.
\begin{figure*}[!t]
	\centering
	\includegraphics[width=2.2in]{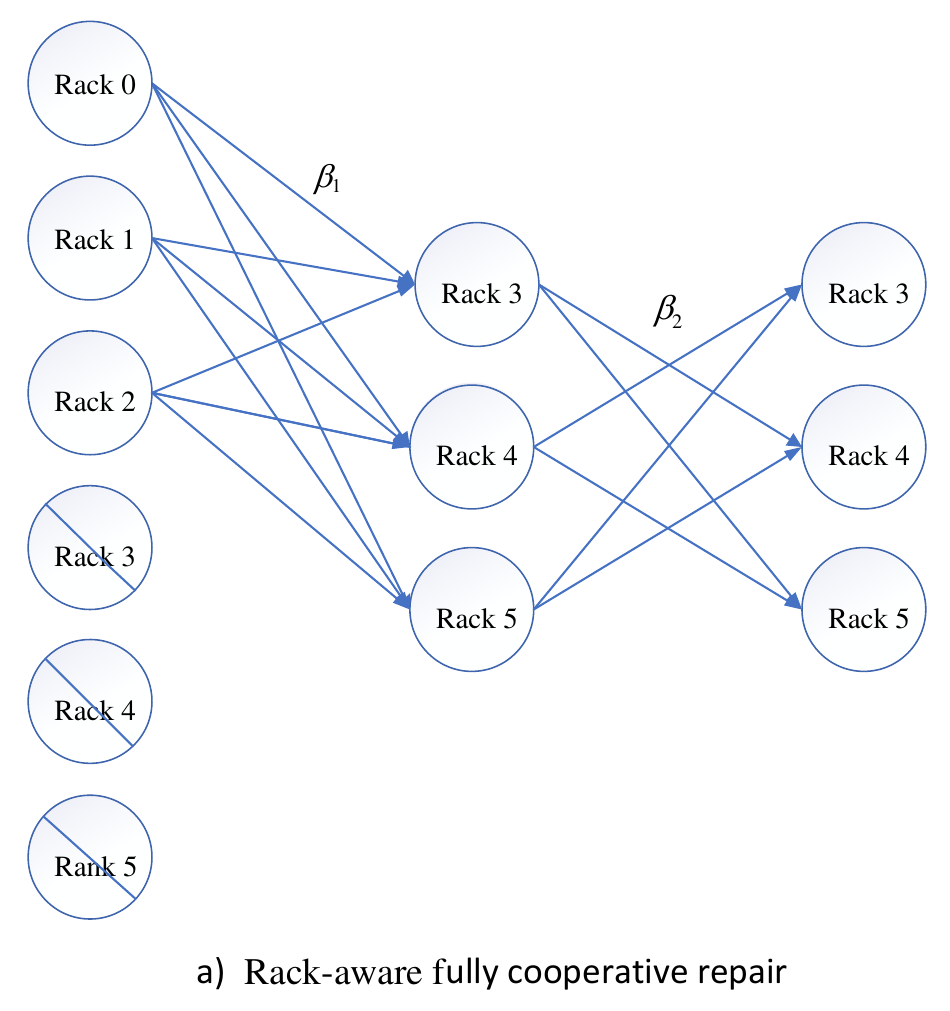}%
	\label{fig4}
	\hfil
	\includegraphics[width=2.2in]{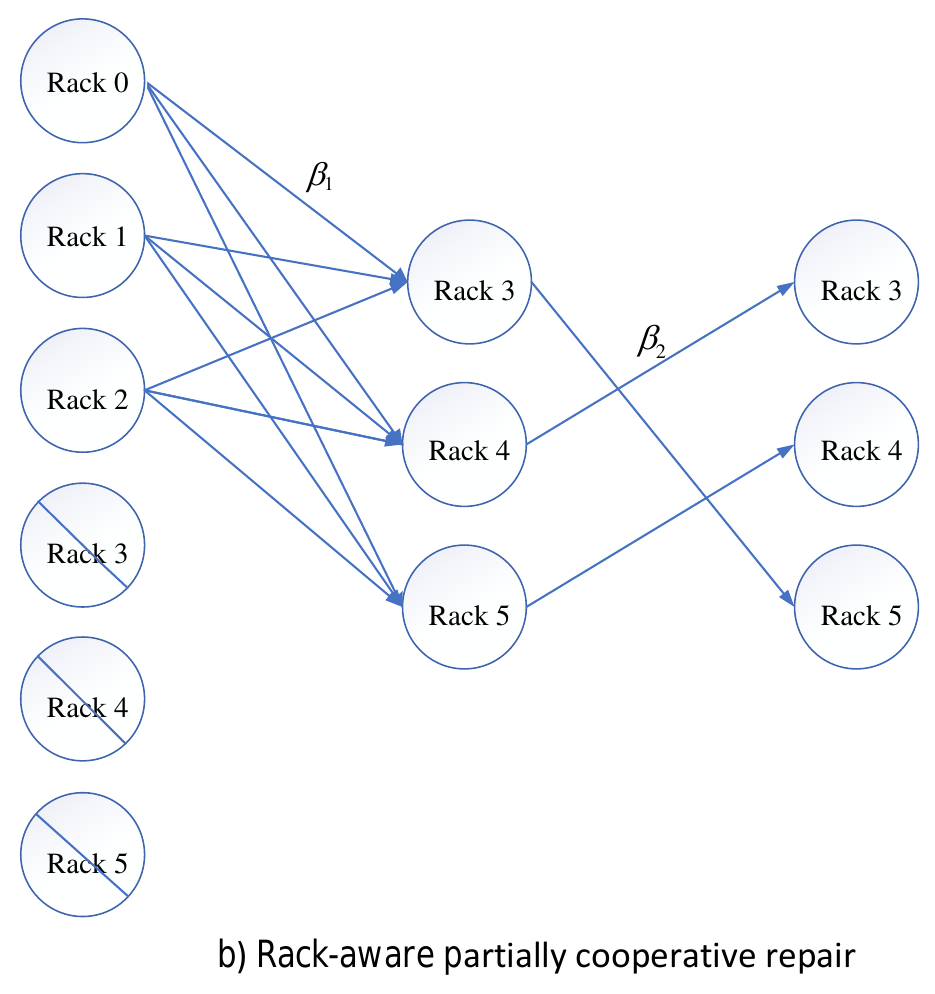}%
	\label{fig4}
	\caption{Rack-aware fully and partially cooperative repair models in a $(12,2)$ MDS array code. The $12$ nodes $(\mathbf{c}_0,\mathbf{c}_1,\ldots,\mathbf{c}_{11})$ are partitioned into $6$ consecutive racks ${0,1,...,5}$, with each rack consisting of exactly $2$ nodes.	 Given that the set of failed nodes is ${\mathbf{c}_6,\mathbf{c}_8,\mathbf{c}_{10}}$, the corresponding set of host racks is ${3,4,5}$. a) Rack-aware fully cooperative repair. Host rack $3$ downloads $\beta_1$ symbols from each helper rack $j$, $j\in[3]$ in the download phase, and then downloads $\beta_2$ symbols from each host rack $i$, $i\in\{4,5\}$. Following these phases, host rack $3$ leverages the symbols generated from the previous two phases, combined with those from its own surviving node $\mathbf{c}_7$, to recover the failed node $\mathbf{c}_6$.
		Host racks $4$ and $5$, the remaining two, are also similar. b) Rack-aware partially cooperative repair. The download phase is the same as a). During the cooperative phase, host rack $3$ downloads $\beta_2$ symbols only from host rack $4$, excluding host rack $5$. Host racks $4$ and $5$ are also similar.
	}
	\label{fig3}
\end{figure*}


The repair bandwidth of the rack-aware partially cooperative repair model is defined as
\begin{equation}\label{e1}
	\gamma(n,k,l,d,h,\delta)=\sum_{p\in[\bar{h}]}\sum_{j\in{\cal R}}\beta_1(i_p,j)+\sum_{p\in[\bar{h}]}\sum_{i'\in\mathcal{S}_p}\beta_2(i_p,i').
\end{equation}Denote the minimum repair bandwidth by $\gamma^{*}(n,k,l,d,h,\delta)$. An $(n,k,l,d,h,\delta)$ MDS array code is called rack-aware MSPCR codes if we can design its rack-aware partially cooperative repair model achieving the minimum repair bandwidth. In this paper,  we will construct $(n,k,l)$ rack-aware MSPCR codes.  

For convenience, the parameters of the array codes used are summarized in Table \ref{tab2}.

\begin{table}[http!]
\renewcommand{\arraystretch}{1.2}
\setlength\tabcolsep{3pt} 
	\centering
	\caption{ Code Parameter
		\label{tab2}
	} 
	\begin{tabular}{|c|c|c|c|c|c|c|c|c|}	
		\hline
		Notation & Meaning & Notation &	 Meaning \\ 
		\hline
		$n$ & code length & $k$ &	 code dimension \\ 		
		\hline
		$d$ &	the number of helper nodes  & $h$ &	the number of helper nodes\\
		\hline
		$r$  &	$n-k$ & $l$ &  the sub-packetization level of the code\\
		\hline
		$\bar{n}$ & the number of racks in the code & $u$ & the size of a rack\\	
		\hline
		$v$ & an integer within $[u]$ & $\bar{k}$ & $(k-v)/u$\\		
		\hline	
		$\bar{r}$ & $\bar{n}-\bar{k}$ &	$\bar{d}$ & the number of helper racks\\		
		\hline	
		$\bar{s}$ & $\bar{d}-\bar{k}+1$ &	$\bar{h}$ &  the number of host racks\\		
		\hline	
		$b$ & the number of failed nodes in a host rack& $\delta$ & cooperative parameter \\		
		\hline	
		$\gamma$ & repair bandwidth & &	\\
		\hline	
	\end{tabular}
\end{table}

\section{Main results}
\label{Sect-main-results}

In this section, we first derive a new lower bound on the repair bandwidth. Based on this bound and employing  the stacking method and the partitioning method respectively, we propose two classes of MDS array codes, each of which contains a rack-aware MSPCR code for $b\in[1,u-v]$ and an asymptotically rack-aware MSPCR code for $b\in[u-v+1,u]$.

Now let us see the new lower bound first. Different from the cut-set bounds proposed in \cite{LO} and \cite{GDL}, we obtain the following lower bound by extremal combinatorics. The proof is presented in Appendix \ref{app-th1}.
\begin{theorem}\label{th1}
	Given  an $(n,k,l,d,h,\delta)$ MDS array code, the following inequality holds for the repair bandwidth  defined in \eqref{e1}
	\begin{equation}\label{e2}
		\gamma(n,k,l,d,h,\delta)\geq\frac{h(\bar{d}+\bar{h}-\delta)l}{\bar{d}-\bar{k}+\bar{h}-\delta+1}.
	\end{equation}
	Moreover, if $\beta_1(i,j)=\beta_1$  and $\beta_2(i,i')=\beta_2$ for any $j\in\mathcal{R}$,
	$i,i'\in\mathcal{F}$ with $i\not=i'$, then the equality holds if and only if $\beta_1=\beta_2=
	\frac{bl}{\bar{d}-\bar{k}+\bar{h}-\delta+1}$.
\end{theorem} 

According to the bound in Theorem \ref{th1}, by stacking multiple instances of an MDS array code, we obtain the first class of rack-aware MSPCR codes as follows, whose proof is shown in Section \ref{section-th2}. 

\begin{theorem}\label{th2}
	For any positive integers  $\bar{n}$, $\bar{k}$, $\bar{d}$, $\bar{h}$ and $\delta$ where  $\bar{d}\in[\bar{k},\bar{n}-\bar{h}]$ and  $\delta\in[1,\bar{h})$, there exist
	\begin{itemize}
		\item  an $(n,k,l=(\bar{s}+\bar{h}-\delta)\bar{s}^{\bar{n}},d,h,\delta)$ rack-aware MSPCR code over $\mathbb{F}$ where $\bar{s}=\bar{d}-\bar{k}+1$, $b\in[1,u-v]$, and $|\mathbb{F}|\geq n\bar{s}+1$  with $u|(|\mathbb{F}|-1)$;
		\item an asymptotically $(n,k,l=(\bar{s}+\bar{h}-\delta)\bar{s}^{\bar{n}},d,h,\delta)$ rack-aware MSPCR code over $\mathbb{F}$ where $\bar{s}=\bar{d}-\bar{k}+1$,  $b\in[u-v+1,u]$, and $|\mathbb{F}|\geq n\bar{s}+1$  with $u|(|\mathbb{F}|-1)$.
	\end{itemize}
\end{theorem} 
When $\bar{d}=\bar{k}+1$, we can use the similar method proposed in \cite{LCT} to further reduce the sub-packetization level of the rack-aware MSPCR codes in Theorem \ref{th2}. Specifically, by dividing the symbols at each node of the code into appropriate groups in the repair process, we obtain the second class of rack-aware MSPCR codes, whose proof can be found in Section \ref{section-th3}.
\begin{theorem} \label{th3}
	For any positive integers  $\bar{n}$, $\bar{k}$, $\bar{d}$, $\bar{h}$ and $\delta$  where $\bar{d}=\bar{k}+1$, $(\bar{h}-\delta+1)|\bar{h}$ and  $\delta\in[1,\bar{h})$ , there exist
	\begin{itemize}
		\item  an $(n,k,l=2^{\bar{n}},d,h,\delta)$ rack-aware MSPCR code over $\mathbb{F}$ where $b\in[1,u-v]$ and $|\mathbb{F}|\geq 2n+1$  with $u|(|\mathbb{F}|-1)$;
		\item  an asymptotically $(n,k,l=2^{\bar{n}},d,h,\delta)$ rack-aware MSPCR code over $\mathbb{F}$ where  $b\in[u-v+1,u]$ and $|\mathbb{F}|\geq 2n+1$  with $u|(|\mathbb{F}|-1)$.
	\end{itemize}
	
\end{theorem}

\begin{remark}
	\begin{itemize}
		\item The sub-packetization level in Theorem \ref{th3} is reduced by a factor of $\bar{h}-\delta+2$ relative to Theorem \ref{th2} when $\bar{d}=\bar{k}+1$.
		
		\item Theorem \ref{th3} is applicable only when $\bar{d}=\bar{k}+1$. So it is meaningful to explore alternative methods for other values of $\bar{d}$ while maintaining a small sub-packetization level.
	\end{itemize}	
\end{remark}

\section{The proof of Theorem \ref{th2}}
\label{section-th2}
Let $\bar{s}=\bar{d}-\bar{k}+1$, $\mathbb{F}$ be a finite field of size
$|\mathbb{F}|\geq n\bar{s}+1$ with $u|(|\mathbb{F}|-1)$, $\xi$ be a primitive element of $\mathbb{F}$,  $\theta$ be an
element of $\mathbb{F}$ with multiplicative order $u$,
$\lambda_{i,j}=\xi^{i\bar{s}+j}$ for $i\in[\bar{n}]$,
$j\in[\bar{s}]$ be $\bar{n}\bar{s}$ distinct elements in $\mathbb{F}$.
By assumption of the rack-aware storage model, we know that the number of helper nodes is $d=\bar{d}u+u-b$. To guarantee $d\geq k$, we need $\bar{k}\leq\bar{d}\leq \bar{n}-\bar{h}$ when $1\leq b\leq u-v$, and $\bar{k}+1\leq\bar{d}\leq \bar{n}-\bar{h}$ when $u-v<b\leq u$. Through modifying the parity check matrix of the code mentioned in \cite{YB1} and then stacking multiple instances of an MDS array code, we construct a class of rack-aware MSPCR  codes. Consider a code $\mathcal{C}=({\bf c}_0,{\bf c}_1,\ldots,{\bf
	c}_{n-1})$ defined by the following parity check equations over $\mathbb{F}$:
\begin{equation}\label{e6}
	\sum\limits_{i=0}^{\bar{n}-1}\sum\limits_{g=0}^{u-1}\mathbf{A}_{i,g}^t{\bf
		c}_{iu+g}=0,
\end{equation}
\noindent where $t\in[r]$, ${\bf c}_{iu+g}=(c_{iu+g,0},c_{iu+g,1},\ldots,c_{iu+g,\bar{s}^{\bar{n}}-1})^\top$ is a column vector of length $\bar{s}^{\bar{n}}$ over $\mathbb{F}$, and
\begin{align}\label{e7}
	&\mathbf{A}_{i,g}=\theta^g\mathbf{A}_i, \\ \nonumber &\mathbf{A}_i=\mathbf{I}_{\bar{s}^{\bar{n}-i-1}}\otimes
	{\rm blkdiag}(\lambda_{i,0}\mathbf{I}_{\bar{s}^i},\ldots,\lambda_{i,\bar{s}-1}\mathbf{I}_{\bar{s}^i}), 
\end{align}

\noindent where $\otimes$ is the Kroncker product. It should be noted that each $\mathbf{A}_{i,g}$ is a diagonal matrix, and the diagonal entry located in the $a$-th row is $\theta^g\lambda_{i,a_i}$. Based on this observation, we are able to rewrite the parity check equations in (\ref{e6}) in the following way.

\begin{equation}\label{e8}
	\sum\limits_{i=0}^{\bar{n}-1}\sum\limits_{g=0}^{u-1}\theta^{gt}\lambda_{i,a_i}^tc_{iu+g,a}=0,
\end{equation}
\noindent where $t\in[r]$,
$a\in[\bar{s}^{\bar{n}}]$, and $a_i$ is the $i$-th term of the 
$\bar{s}$-ary expansion of $a=(a_{\bar{n}-1},\ldots,a_1,a_0)$.

Note that $\lambda_{i,a_i}=\xi^{i\bar{s}+a_i}$ and $\theta^u=1$, it is easy to see that 
$\theta^{g}\lambda_{i,a_i},g\in[u],i\in[\bar{n}],a_i\in[\bar{s}]$ are distinct.
Similar to the discussion in Theorem IV.1 in \cite{YB1}, we know
that the code $\mathcal{C}$ given in (\ref{e6}) has the MDS property.
Finally, our rack-aware MSPCR code $\widetilde{\mathcal{C}}=(\tilde{\mathbf{c}}_0,\tilde{\mathbf{c}}_1,\ldots,\tilde{\mathbf{c}}_{n-1})$ is defined by the following parity check equations over $\mathbb{F}$:
\begin{equation}\label{e9}
	\sum\limits_{i=0}^{\bar{n}-1}\sum\limits_{g=0}^{u-1}(\widetilde{\mathbf{A}}_{i,g})^t\tilde{\mathbf{c}}_{iu+g}=0,
\end{equation}
\noindent where
$\widetilde{\mathbf{A}}_{i,g}=\mathbf{I}_{\bar{s}+\bar{h}-\delta}\otimes \mathbf{A}_{i,g}$ and
$\tilde{\mathbf{c}}_{iu+g}=(({\bf c}_{iu+g}^{(0)})^\top,\ldots,({\bf c}_{iu+g}^{(\bar{s}+\bar{h}-\delta-1)})^\top)^\top$
of length $l=(\bar{s}+\bar{h}-\delta)\bar{s}^{\bar{n}}$. In other
words, we generate $\bar{s}+\bar{h}-\delta$ instances of the MDS array code $\mathcal{C}$ defined in (\ref{e6}). The code $\widetilde{\mathcal{C}}$
defined in (\ref{e9}) has the MDS property since the code $\mathcal{C}$ is an MDS array
code and $\widetilde{\mathbf{A}}_{i,g}=\mathbf{I}_{\bar{s}+\bar{h}-\delta}\otimes
\mathbf{A}_{i,g}$. Similarly, we can rewrite the parity check equations in
(\ref{e9}) as
\begin{align}\label{e10}
	\sum\limits_{i=0}^{\bar{n}-1}\sum\limits_{g=0}^{u-1}\theta^{gt}\lambda_{i,a_i}^tc_{iu+g,a}^{(y)}=0
\end{align}
\noindent for all $t\in[r]$,
$a\in[\bar{s}^{\bar{n}}]$, and
$y\in[\bar{s}+\bar{h}-\delta]$.

In this section,  for any two racks $i_p\in{\cal F}$ and $j\in[\bar{n}]\setminus\{i_p\}$,  we define
\begin{align*}
	H_{i_p,j}(a,m)=\sum\limits_{g=0}^{u-1}\theta^{gm}c_{ju+g,a}^{(p)}+\sum\limits_{g=0}^{u-1}\theta^{gm}c_{ju+g,a(i_p,a_{i_p}\oplus1)}^{(p+1)}+
	\ldots+\sum\limits_{g=0}^{u-1}\theta^{gm}c_{ju+g,a(i_p,a_{i_p}\oplus(\bar{s}-1))}^{(p+\bar{s}-1)},
\end{align*}
where $a\in[\bar{s}^{\bar{n}}]$, $m\in[b]$ and the superscript operation is addition modulo
$\bar{s}+\bar{h}-\delta$.

In general, the partially cooperative
repair model for code $\widetilde{\mathcal{C}}$ generated by the
parity check equations (\ref{e10}) with $h$ failed nodes can be given in the
following two phases.

$\bullet$ Download phase: When $b\in[1,u-v]$, each host rack $i_p\in {\cal F}$ downloads $H_{i_p,j}(a,m), a\in[\bar{s}^{\bar{n}}], m\in[b]$ from all the helper racks $j\in{\cal R}$. When $b\in[u-v+1,u]$,  each host rack $i_p\in {\cal F}$ downloads $H_{i_p,j}(a,m),a\in[\bar{s}^{\bar{n}}], m\in[u-v]$ from all the helper racks $j\in{\cal R}$, and additionally downloads $H_{i_p,j}(a,m),a\in[\bar{s}^{\bar{n}}], m\in[u-v,b)$ from all the helper racks $j\in{\cal R}\cup\{j'\}$, where $j'\in[\bar{n}]\setminus(\mathcal{R}\cup\mathcal{F})$.

$\bullet$ Cooperative phase: Each host rack $i_q\in \mathcal{S}_p$
transfers $H_{i_q,i_p}(a,m)$, $a\in[\bar{s}^{\bar{n}}]$, $m\in[b]$ to host rack $i_p$,
where

\begin{align}\label{e11}
\mathcal{S}_p=\left\{\begin{array}{cc}
		\{i_q:q\in[p+1,p+\bar{h}-\delta]\},  {\rm if} \  0\leq p<\delta,\\		
		\{i_q:q\in[\delta-1,p-1]\cup[p+1,\bar{h}-1]\},  {\rm if} \ \delta\leq	p<\bar{h}-\delta,\\
		\{i_q:q\in[p-\bar{h}+\delta,p-1]\},  {\rm if} \ \text{ max}\{\delta,\bar{h}-\delta\}\leq p\leq \bar{h}-1.
	\end{array}
	\right.
\end{align}

For $p\in[\bar{h}]$, note that $|\mathcal{S}_p|=\bar{h}-\delta$ and $\mathcal{S}_p\subseteq\mathcal{F}\setminus\{i_p\}$. 
The selection of the set $\mathcal{S}_p$ must satisfy the following condition: the subscripts of the elements in $\mathcal{S}_p\cup{i_p}$ form $\bar{h}-\delta+1$ consecutive integers under modulo $\bar{h}$. For example, let $\bar{s}=3$, $\bar{h}=6$, $\delta=3$, $\bar{d}=5$, then $\bar{s}+\bar{h}-\delta=6$.
Let us take $p=0$ and $\mathcal{S}_0 =\{i_5, i_1, i_2\}$. Then, $\mathcal{S}_0 \cup {i_0} = \{i_5, i_0, i_1, i_2\}$, whose element subscripts are $5, 0, 1, 2$. These subscripts form a set of $4$ consecutive integers modulo $6$. The set $\mathcal{S}_p$ $(p\in[\bar{h}])$  defined in \eqref{e11} is merely one of the possible choices.

The following results show that some linear combinations of
the symbols  can be recovered by the host racks during the download phase, and such combinations will be employed in the partially cooperative phase.

\begin{lemma} \label{lm1}
	Following the notations introduced above, we have the following results:
	
	1) for any given integer $b\in[1,u-v]$,  by downloading
	$H_{i_p,j}(a,m)$, $m\in[b]$ from all the helper racks $j\in{\cal R}$, host rack $i_p$ can recover
	
	$\bullet$	$\sum\limits_{g=0}^{u-1}\theta^{gm}c_{i_pu+g,a}^{(p)}$,
	$\sum\limits_{g=0}^{u-1}\theta^{gm}c_{i_pu+g,a}^{(p+1)}$,\ldots,
	$\sum\limits_{g=0}^{u-1}\theta^{gm}c_{i_pu+g,a}^{(p+\bar{s}-1)}$;
	
	$\bullet$	$H_{i_p,i}(a,m)$, 
	
	\noindent	where $\forall i\in{\cal F}\setminus\{i_p\}$, $a\in[\bar{s}^{\bar{n}}]$.
	
	\indent\indent	2) for any given integer $b\in[u-v+1,u]$,  by downloading
	$H_{i_p,j}(a,m)$, $m\in[u-v,b)$ from all the helper racks $j\in{\cal R}\cup\{j'\}$ with
	$j'\in[\bar{n}]\setminus(\mathcal{R}\cup\mathcal{F})$, host rack $i_p$ can recover
	
	$\bullet$ $\sum\limits_{g=0}^{u-1}\theta^{gm}c_{i_pu+g,a}^{(p)}$,
	$\sum\limits_{g=0}^{u-1}\theta^{gm}c_{i_pu+g,a}^{(p+1)}$,\ldots,
	$\sum\limits_{g=0}^{u-1}\theta^{gm}c_{i_pu+g,a}^{(p+\bar{s}-1)}$;
	
	$\bullet$ $H_{i_p,i}(a,m)$, 
	
	\noindent where $\forall i\in{\cal F}\setminus\{i_p\}$, $a\in[\bar{s}^{\bar{n}}]$.
\end{lemma}
\begin{proof} Let $t=\bar{t}u+m\in[r]$ where $m\in[u]$, for each integer $m\in[0,u-v)$ we have $\bar{t}\in[\bar{r}]$, otherwise $\bar{t}\in[\bar{r}-1]$. We only prove  conclusion 1),  conclusion 2) can be proved similarly. Let ${\cal U}=[\bar{n}]\setminus({\cal F}\cup{\cal R})=\{z_0,z_1,\ldots,z_{\bar{n}-\bar{h}-\bar{d}-1}\}$ be the index set of the unconnected racks. Since $\theta^u=1$, then from (\ref{e10}) we can get
	\begin{center}
		$\lambda_{i_p,a_{i_p}\oplus
			x}^{\bar{t}u+m}\sum\limits_{g=0}^{u-1}\theta^{gm}c_{i_pu+g,a(i_p,a_{i_p}\oplus
			x)}^{(y)}+\sum\limits_{i\in{\cal F}\setminus\{i_p\}}\lambda_{i,a_{i}}^{\bar{t}u+m}\sum\limits_{g=0}^{u-1}\theta^{gm}c_{iu+g,a(i_p,a_{i_p}\oplus x)}^{(y)}$
	\end{center}
	\begin{center}
		$+\sum\limits_{j\in{\cal R}}\lambda_{j,a_{j}}^{\bar{t}u+m}\sum\limits_{g=0}^{u-1}\theta^{gm}c_{ju+g,a(i_p,a_{i_p}\oplus
			x)}^{(y)}+\sum\limits_{z\in{\cal U}}\lambda_{z,a_{z}}^{\bar{t}u+m}\sum\limits_{g=0}^{u-1}\theta^{gm}c_{zu+g,a(i_p,a_{i_p}\oplus
			x)}^{(y)}=0$
	\end{center}
	\noindent for any $a\in[\bar{s}^{\bar{n}}]$, $x\in[\bar{s}]$,
	$y\in[\bar{s}+\bar{h}-\delta]$, $\bar{t}\in[\bar{r}]$ and $m\in[b]$, which implies
	\begin{align*}
		\sum\limits_{y=p}^{p+\bar{s}-1}(\lambda_{i_p,a_{i_p}\oplus
			(y-p)}^{\bar{t}u+m}\sum\limits_{g=0}^{u-1}\theta^{gm}c_{i_pu+g,a(i_p,a_{i_p}\oplus
			(y-p))}^{(y)}+\sum\limits_{i\in{\cal F}\setminus\{i_p\}}\lambda_{i,a_{i}}^{\bar{t}u+m}\sum\limits_{g=0}^{u-1}\theta^{gm}c_{iu+g,a(i_p,a_{i_p}\oplus
			(y-p))}^{(y)}\\
		+\sum\limits_{j\in{\cal R}}\lambda_{j,a_{j}}^{\bar{t}u+m}\sum\limits_{g=0}^{u-1}\theta^{gm}c_{ju+g,a(i_p,a_{i_p}\oplus
			(y-p))}^{(y)}+\sum\limits_{z\in{\cal U}}\lambda_{z,a_{z}}^{\bar{t}u+m}\sum\limits_{g=0}^{u-1}\theta^{gm}c_{zu+g,a(i_p,a_{i_p}\oplus
			(y-p))}^{(y)})=0.
	\end{align*}
	\noindent This equation can be rewritten as
	\begin{center}
		$\sum\limits_{y=p}^{p+\bar{s}-1}\lambda_{i_p,a_{i_p}\oplus
			(y-p)}^{\bar{t}u+m}\sum\limits_{g=0}^{u-1}\theta^{gm}c_{i_pu+g,a(i_p,a_{i_p}\oplus
			(y-p))}^{(y)}+\sum\limits_{i\in{\cal F}\setminus\{i_p\}}\lambda_{i,a_{i}}^{\bar{t}u+m}\sum\limits_{y=p}^{p+\bar{s}-1}\sum\limits_{g=0}^{u-1}\theta^{gm}c_{iu+g,a(i_p,a_{i_p}\oplus
			(y-p))}^{(y)}$
	\end{center}
	\begin{equation}\label{e12}
		+\sum\limits_{z\in{\cal U}}\lambda_{z,a_{z}}^{\bar{t}u+m}\sum\limits_{y=p}^{p+\bar{s}-1}\sum\limits_{g=0}^{u-1}\theta^{gm}c_{zu+g,a(i_p,a_{i_p}\oplus
			(y-p))}^{(y)}=-\sum\limits_{j\in{\cal R}}\lambda_{j,a_{j}}^{\bar{t}u+m}\sum\limits_{y=p}^{p+\bar{s}-1}\sum\limits_{g=0}^{u-1}\theta^{gm}c_{ju+g,a(i_p,a_{i_p}\oplus
			(y-p))}^{(y)},
	\end{equation}
	\noindent where $a\in[\bar{s}^{\bar{n}}]$, $\bar{t}\in[\bar{r}]$ and $m\in[b]$. Let
	
	\begin{center}
		${\bf v}_{i_p,a,m}^\top=[\sum\limits_{g=0}^{u-1}\theta^{gm}c_{i_pu+g,a}^{(p)},
		\sum\limits_{g=0}^{u-1}\theta^{gm}c_{i_pu+g,a(i_p,a_{i_p}\oplus1)}^{(p+1)},\ldots,
		\sum\limits_{g=0}^{u-1}\theta^{gm}c_{i_pu+g,a(i_p,a_{i_p}\oplus(\bar{s}-1))}^{(p+\bar{s}-1)},H_{i_p,i_0}(a,m),$
	\end{center}
	\begin{center}
		\noindent$\ldots, H_{i_p,i_{p-1}}(a,m),
		H_{i_p,i_{p+1}}(a,m),\ldots,H_{i_p,i_{\bar{h}-1}}(a,m),H_{i_p,z_0}(a,m),\ldots,H_{i_p,z_{\bar{n}-\bar{h}-\bar{d}-1}}(a,m)]\in\mathbb{F}^{\bar{r}}$
	\end{center}
	\noindent and
	\begin{center}
		${\bf u}_{i_p,a,m}^\top=[-H_{i_p,j_0}(a,m),-H_{i_p,j_1}(a,m),\ldots,-H_{i_p,j_{\bar{d}-1}}(a,m)]\in \mathbb{F}^{\bar{d}}$
	\end{center}
	\noindent for ${\cal R}=\{j_0,j_1,\ldots,j_{\bar{d}-1}\}$.
	For any $a\in[\bar{s}^{\bar{n}}]$ and  $m\in[b]$,  \eqref{e12} can be rewritten as
	\begin{center}
		${\bf M}_{i_p,a,m}{\bf v}_{i_p,a,m}={\bf M}'_{i_p,a,m}{\bf u}_{i_p,a,m}$,
	\end{center}
	\noindent where
	
	{\footnotesize 
		\[{\bf M}_{i_p,a,m}=
		\left(
		\begin{array}{ccccccc}
			\lambda_{i_p,a_{i_p}}^{m}   \ldots 
			\lambda_{i_p,a_{i_p}\oplus(\bar{s}-1)}^{m} &
			\lambda_{i_0,a_{i_0}}^{m}  \ldots \lambda_{i_{p-1},a_{i_{p-1}}}^{m}
			& \lambda_{i_{p+1},a_{i_{p+1}}}^{m}  \ldots
			\lambda_{i_{\bar{h}-1},a_{i_{\bar{h}-1}}}^{m} &
			\lambda_{z_0,a_{z_0}}^{m} \ldots  \lambda_{z_{\bar{n}-\bar{h}-\bar{d}-1},a_{z_{\bar{n}-\bar{h}-\bar{d}-1}}}^{m}\\
			\lambda_{i_p,a_{i_p}}^{u+m}  \ldots 
			\lambda_{i_p,a_{i_p}\oplus(\bar{s}-1)}^{u+m} &
			\lambda_{i_0,a_{i_0}}^{u+m}  \ldots 
			\lambda_{i_{p-1},a_{i_{p-1}}}^{u+m} &
			\lambda_{i_{p+1},a_{i_{p+1}}}^{u+m}  \ldots 
			\lambda_{i_{\bar{h}-1},a_{i_{\bar{h}-1}}}^{u+m} &
			\lambda_{z_0,a_{z_0}}^{u+m}  \ldots  \lambda_{z_{\bar{n}-\bar{h}-\bar{d}-1},a_{z_{\bar{n}-\bar{h}-\bar{d}-1}}}^{u+m}\\
			\vdots \indent \indent \indent \vdots \indent  & \vdots \indent \indent \indent \vdots \indent & \vdots \indent\indent \indent \vdots \indent & \vdots \indent \indent \indent \indent \vdots \indent\indent\indent\indent \indent \\
			\lambda_{i_p,a_{i_p}}^{(\bar{r}-1)u+m}  \ldots 
			\lambda_{i_p,a_{i_p}\oplus(\bar{s}-1)}^{(\bar{r}-1)u+m} &
			\lambda_{i_0,a_{i_0}}^{(\bar{r}-1)u+m}  \ldots 
			\lambda_{i_{p-1},a_{i_{p-1}}}^{(\bar{r}-1)u+m} &
			\lambda_{i_{p+1},a_{i_{p+1}}}^{(\bar{r}-1)u+m}  \ldots 
			\lambda_{i_{\bar{h}-1},a_{i_{\bar{h}-1}}}^{(\bar{r}-1)u+m} &
			\lambda_{z_0,a_{z_0}}^{(\bar{r}-1)u+m} \ldots \lambda_{z_{\bar{n}-\bar{h}-\bar{d}-1},a_{z_{\bar{n}-\bar{h}-\bar{d}-1}}}^{(\bar{r}-1)u+m}\\
		\end{array}
		\right),
		\]
	}

	\noindent and
	
	{\footnotesize
		\[ {\bf M}'_{i_p,a,m}=
		\left(
		\begin{array}{ccccccc}
			\lambda_{j_0,a_{j_0}}^{m} & \lambda_{j_1,a_{j_1}}^{m} & \cdots &
			\lambda_{j_{\bar{d}-1},a_{j_{\bar{d}-1}}}^{m}\\
			\lambda_{j_0,a_{j_0}}^{u+m} & \lambda_{j_1,a_{j_1}}^{u+m} & \cdots &
			\lambda_{j_{\bar{d}-1},a_{j_{\bar{d}-1}}}^{u+m}\\
			\vdots & \vdots & \ddots &
			\vdots\\
			\lambda_{j_0,a_{j_0}}^{(\bar{r}-1)u+m} &
			\lambda_{j_1,a_{j_1}}^{(\bar{r}-1)u+m} & \cdots &
			\lambda_{j_{\bar{d}-1},a_{j_{\bar{d}-1}}}^{(\bar{r}-1)u+m}\\
		\end{array}
		\right).
		\]
	}
	
Note that for any $a\in[\bar{s}^{\bar{n}}]$ and $m\in[b]$, host rack $i_p$ downloads $H_{i_p,j}(a,m)$ from each helper rack $j\in{\cal R}$, thus ${\bf u}_{i_p,a,m}$ is
obtained. The matrix ${\bf M}_{i_p,a,m}$ can be rewritten as
	${\bf M}_1$diag$(\lambda_{i_p,a_{i_p}}^{m},\ldots,\lambda_{i_p,a_{i_p}\oplus(\bar{s}-1)}^{m},$ $\ldots,
	\lambda_{z_0,a_{z_0}}^{m},\ldots,\lambda_{z_{\bar{n}-\bar{h}-\bar{d}-1},a_{z_{\bar{n}-\bar{h}-\bar{d}-1}}}^{m})$,
	where
	
	{\footnotesize 
		\[{\bf M}_1=
		\left(
		\begin{array}{ccccccc}
			1 \indent  \ldots \indent 1\indent\indent \indent\indent  & 	1 \indent  \ldots \indent 1\indent\indent \indent\indent & 	1 \indent \indent \ldots \indent \indent 1\indent\indent \indent\indent &\indent 1 \indent  \ldots \indent \indent 1 \indent\indent\indent\indent\indent\indent\indent \indent \\
			\lambda_{i_p,a_{i_p}}^{u}  \ldots 
			\lambda_{i_p,a_{i_p}\oplus(\bar{s}-1)}^{u} &
			\lambda_{i_0,a_{i_0}}^{u}  \ldots 
			\lambda_{i_{p-1},a_{i_{p-1}}}^{u} &
			\lambda_{i_{p+1},a_{i_{p+1}}}^{u}  \ldots 
			\lambda_{i_{\bar{h}-1},a_{i_{\bar{h}-1}}}^{u} &
			\lambda_{z_0,a_{z_0}}^{u}  \ldots  \lambda_{z_{\bar{n}-\bar{h}-\bar{d}-1},a_{z_{\bar{n}-\bar{h}-\bar{d}-1}}}^{u}\\
			\vdots \indent \indent \indent \vdots \indent\indent  & \vdots \indent \indent \indent \vdots \indent & \vdots \indent\indent \indent \vdots \indent & \vdots \indent \indent \indent \indent \vdots \indent\indent\indent\indent \indent \\
			\lambda_{i_p,a_{i_p}}^{(\bar{r}-1)u}  \ldots 
			\lambda_{i_p,a_{i_p}\oplus(\bar{s}-1)}^{(\bar{r}-1)u} &
			\lambda_{i_0,a_{i_0}}^{(\bar{r}-1)u}  \ldots 
			\lambda_{i_{p-1},a_{i_{p-1}}}^{(\bar{r}-1)u} &
			\lambda_{i_{p+1},a_{i_{p+1}}}^{(\bar{r}-1)u}  \ldots 
			\lambda_{i_{\bar{h}-1},a_{i_{\bar{h}-1}}}^{(\bar{r}-1)u} &
			\lambda_{z_0,a_{z_0}}^{(\bar{r}-1)u} \ldots \lambda_{z_{\bar{n}-\bar{h}-\bar{d}-1},a_{z_{\bar{n}-\bar{h}-\bar{d}-1}}}^{(\bar{r}-1)u}\\
		\end{array}
		\right).
		\]
	}
	
	Since $\lambda_{i,j}=\xi^{i\bar{s}+j}$, then we have $\lambda_{i,j}^u\neq
	\lambda_{i',j'}^u$ for all $i,i'\in[\bar{n}]$, $j,j'\in[\bar{s}]$
	with $(i,j)\neq(i',j')$,  the Vandermonde matrix ${\bf M}_1$ are
	invertible for each $a\in[0,\bar{s}^{\bar{n}}]$. It follows that the
	matrix ${\bf M}_{i_p,a,m}$ is also invertible, therefore
	
	\begin{center}
		${\bf v}_{i_p,a,m}={\bf M}_{i_p,a,m}^{-1}{\bf M}'_{i_p,a,m}{\bf u}_{i_p,a,m}$,
		$a\in[0,\bar{s}^{\bar{n}}], m\in[b]$.
	\end{center}
	\noindent In other words,  host rack $i_p$ can recover
	\begin{align*}
		&\{\sum\limits_{g=0}^{u-1}\theta^{gm}c_{i_pu+g,a(i_p,a_{i_p}\oplus(y-p))}^{(y)}:a\in[\bar{s}^{\bar{n}}],y\in[p,p+\bar{s}),m\in[b]\}\\
		&=\{\sum\limits_{g=0}^{u-1}\theta^{gm}c_{i_pu+g,a}^{(y)}:a\in[\bar{s}^{\bar{n}}],
		y\in[p,p+\bar{s}),m\in[b]\},
	\end{align*}
	\noindent and
	\begin{align*}
		H_{i_p,i}(a,m)&=\sum\limits_{g=0}^{u-1}\theta^{gm}c_{iu+g,a}^{(p)}+\sum\limits_{g=0}^{u-1}\theta^{gm}c_{iu+g,a(i_p,a_{i_p}\oplus1)}^{(p+1)}+
		\ldots+\sum\limits_{g=0}^{u-1}\theta^{gm}c_{iu+g,a(i_p,a_{i_p}\oplus(\bar{s}-1))}^{(p+\bar{s}-1)}, \\ &\forall i\in{\cal F}\setminus\{i_p\},
		a\in[\bar{s}^{\bar{n}}],m\in[b]. 
	\end{align*}
\end{proof}

\begin{lemma} \label{lm2}
	Based on the data
	$\{\sum\limits_{g=0}^{u-1}\theta^{gm}c_{i_pu+g,a}^{(y)}:a\in[\bar{s}^{\bar{n}}],
	y\in[p,p+\bar{s}),m\in[b]\}$ from lemma \ref{lm1}, $b\in[1,u]$, some linear
	combinations of symbols in the $p$-th host rack
	$\{\sum\limits_{g=0}^{u-1}\theta^{gm}c_{i_pu+g,a}^{(y)}:a\in[\bar{s}^{\bar{n}}],
	y\in[\bar{s}+\bar{h}-\delta],m\in[b]\}$ can be recovered if it downloads
	the data $H_{i_q,i_p}(a,m)$, $a\in[\bar{s}^{\bar{n}}],m\in[b]$ from $\bar{h}-\delta$ host racks $i_q\in \mathcal{S}_p$.
\end{lemma}
\begin{proof} We only prove that the conclusion holds for $\mathcal{S}_p$ defined in $(\ref{e11})$. For other selection methods of $\mathcal{S}_p$ ($p\in[\bar{h}]$), the same approach can be used for the proof. We prove it in three cases, depending on the value of $p$. We prove it in three cases, depending on the value of	$p$.
	
	Case 1. When $0\leq p<\delta$, we argue by induction on
	$q=p+1,p+2,\ldots,p+\bar{h}-\delta$. For the case of $q=p+1$, based on
	the date
	$\{H_{i_q,i_p}(a,m)=H_{i_{p+1},i_p}(a,m):a\in[\bar{s}^{\bar{n}}],m\in[b]\}$
	and
	$\{\sum\limits_{g=0}^{u-1}\theta^{gm}c_{i_pu+g,a}^{(y)}:a\in[\bar{s}^{\bar{n}}],
	y\in[p,p+\bar{s}),m\in[b]\}$ from Lemma \ref{lm1}, we can obtain	
	\begin{align*}
		\sum\limits_{g=0}^{u-1}\theta^{gm}c_{i_pu+g,a(i_q,a_{i_q}\oplus(\bar{s}-1))}^{(q+\bar{s}-1)}&=
		\sum\limits_{g=0}^{u-1}\theta^{gm}c_{i_pu+g,a(i_{p+1},a_{i_{p+1}}\oplus(\bar{s}-1))}^{(p+\bar{s})}\\
		&=H_{i_{p+1},i_p}(a,m)-(\sum\limits_{g=0}^{u-1}\theta^{gm}c_{i_pu+g,a}^{(p+1)}+
		\ldots+\sum\limits_{g=0}^{u-1}\theta^{gm}c_{i_pu+g,a(i_{p+1},a_{i_{p+1}}\oplus(\bar{s}-2)}^{(p+\bar{s}-1)}),
	\end{align*}
	\noindent i.e., we can recover
	$\{\sum\limits_{g=0}^{u-1}\theta^{gm}c_{i_pu+g,a}^{(p+\bar{s})}:a\in[\bar{s}^{\bar{n}}],m\in[b]\}$.
	
	Suppose that the conclusion holds for $q=p+q_0$, i.e., we 
	recover
	$\{\sum\limits_{g=0}^{u-1}\theta^{gm}c_{i_pu+g,a}^{(y)}:a\in[\bar{s}^{\bar{n}}],
	y\in[p+q_0,p+q_0+\bar{s}),m\in[b]\}$ for $0\leq q_0< \bar{h}-\delta$. For
	the case of $q=p+q_0+1$, from $H_{i_q,i_p}(a,m)$ we obtain
	\begin{align*}
		\sum\limits_{g=0}^{u-1}\theta^{gm}c_{i_pu+g,a(i_{q},a_{i_q}\oplus(\bar{s}-1))}^{(p+q_0+\bar{s})}&=\sum\limits_{g=0}^{u-1}\theta^{gm}c_{i_pu+g,a(i_{q},a_{i_{q}}\oplus(\bar{s}-1))}^{(q+\bar{s}-1)}\\
		&=H_{i_{q},i_{p}}(a,m)-(\sum\limits_{g=0}^{u-1}\theta^{gm}c_{i_pu+g,a}^{(p+q_0+1)}+
		\ldots+\sum\limits_{g=0}^{u-1}\theta^{gm}c_{i_pu+g,a(i_{q},a_{i_{q}}\oplus(\bar{s}-2))}^{(p+q_0+\bar{s}-1)}),
	\end{align*}
	\noindent which implies 
	\begin{center}
		$\displaystyle\{\sum\limits_{g=0}^{u-1}\theta^{gm}c_{i_pu+g,a(i_q,a_{i_q}\oplus(\bar{s}-1))}^{(p+q_0+\bar{s})}=
		\sum\limits_{g=0}^{u-1}\theta^{gm}c_{i_pu+g,a}^{(p+q_0+\bar{s})}:a\in[\bar{s}^{\bar{n}}],m\in[b]\}$
	\end{center}
	\noindent by the induction assumption. Therefore, by induction we are able to recover
	\begin{align*}
		&\displaystyle\{\sum\limits_{g=0}^{u-1}\theta^{gm}c_{i_pu+g,a}^{(y)}:a\in[\bar{s}^{\bar{n}}],y\in[p,p+\bar{s}+\bar{h}-\delta),m\in[b]\}\\
		=&\{\sum\limits_{g=0}^{u-1}\theta^{gm}c_{i_pu+g,a}^{(y)}:a\in[\bar{s}^{\bar{n}}],y\in[\bar{s}+\bar{h}-\delta],m\in[b]\},
	\end{align*}
	\noindent since the superscript operation 
	is addition modulo $\bar{s}+\bar{h}-\delta$.
	
	Case 2. When ${\rm max}\{\delta,\bar{h}-\delta\}\leq p\leq \bar{h}-1$, we argue by
	induction on $q=p-1,p-2,\ldots,p-\bar{h}+\delta$. For the case of
	$q=p-1$, based on the date
	$\{H_{i_q,i_p}(a,m)=H_{i_{p-1},i_p}(a,m):a\in[\bar{s}^{\bar{n}}],m\in[b]\}$
	and
	$\{\sum\limits_{g=0}^{u-1}\theta^{gm}c_{i_pu+g,a}^{(y)}:a\in[\bar{s}^{\bar{n}}],
	y\in[p,p+\bar{s}),m\in[b]\}$ from Lemma \ref{lm1}, we can obtain
	
	$\displaystyle\sum\limits_{g=0}^{u-1}\theta^{gm}c_{i_pu+g,a}^{(p-1)}=H_{i_{p-1},i_p}(a,m)-(\sum\limits_{g=0}^{u-1}\theta^{gm}c_{i_pu+g,a(i_{p-1},a_{i_{p-1}}\oplus1)}^{(p)}+
	\ldots+\sum\limits_{g=0}^{u-1}\theta^{gm}c_{i_pu+g,a(i_{p-1},a_{i_{p-1}}\oplus(\bar{s}-1)}^{(p+\bar{s}-2)}), a\in[\bar{s}^{\bar{n}}]$.
	
	Suppose that the conclusion holds for $q=p-q_0$, i.e., we recover
	$\{\sum\limits_{g=0}^{u-1}\theta^{gm}c_{i_pu+g,a}^{(y)}:a\in[\bar{s}^{\bar{n}}],
	y\in[p-q_0,p-q_0+\bar{s}),m\in[b]\}$ for $0\leq q_0< \bar{h}-\delta$. For
	the case of $q=p-q_0-1$, from $H_{i_q,i_p}(a,m)$ we can obtain
	\begin{center}
		$\displaystyle\sum\limits_{g=0}^{u-1}\theta^{gm}c_{i_pu+g,a}^{(p-q_0-1)}=H_{i_{q},i_{p}}(a,m)-(\sum\limits_{g=0}^{u-1}\theta^{gm}c_{i_pu+g,a(i_q,a_{i_q}\oplus1)}^{(p-q_0)}+
		\ldots+\sum\limits_{g=0}^{u-1}\theta^{gm}c_{i_pu+g,a(i_{q},a_{i_{q}}\oplus(\bar{s}-1)}^{(p-q_0+\bar{s}-2)})$.
	\end{center}
	
	Therefore, by induction we are able to recover
	\begin{align*}
		&\{\sum\limits_{g=0}^{u-1}\theta^{gm}c_{i_pu+g,a}^{(y)}:a\in[\bar{s}^{\bar{n}}],
		y\in\{p-\bar{h}+\delta,p-\bar{h}+\delta+1,\ldots,p+\bar{s}-1\},m\in[b]\}\\
		=&\{\sum\limits_{g=0}^{u-1}\theta^{gm}c_{i_pu+g,a}^{(y)}:a\in[\bar{s}^{\bar{n}}],y\in[\bar{s}+\bar{h}-\delta],m\in[b]\}
	\end{align*}
	\noindent since the superscript operation 
	is addition modulo $\bar{s}+\bar{h}-\delta$.
	
	Case 3. When $\delta\leq p<\bar{h}-\delta$,  the proof is analogous to those of Case 1 for $q\in[p+1,\bar{h}-1]$  and Case 2 for $q\in[\delta-1,p-1]$, so we omit the proof. 
\end{proof}

Now we are in a position to prove Theorem \ref{th2}.

{\bf The Proof of Theroem \ref{th2}: } We first recover the $b=h/\bar{h}$ failed nodes
in host rack $i_p (p=0,1,\ldots,\bar{h}-1)$. Assume that the index
set of the $b$ failed nodes in rack $i_p$ is
${\cal I}=\{g_0,g_1,\ldots,g_{b-1}\}$ and ${\cal J}=[u]\setminus{\cal I}$. For any $a\in[\bar{s}^{\bar{n}}]$, $y\in[\bar{s}+\bar{h}-\delta]$ and $m\in[b]$, by Lemma \ref{lm2} we know that the data
$\Delta_{m,a}^{(y)}=\sum\limits_{g=0}^{u-1}\theta^{gm}c_{i_pu+g,a}^{(y)}$,
thus
\begin{align}\label{e18-1}
	\displaystyle\sum\limits_{g\in{\cal I}}\theta^{gm}c_{i_pu+g,a}^{(y)}=\Delta_{m,a}^{(y)}-\sum\limits_{g\in{\cal J}}\theta^{gm}c_{i_pu+g,a}^{(y)}.
\end{align}
\noindent When $m$ runs over $\{0,1,\ldots,b-1\}$, \eqref{e18-1} can be rewritten as
\begin{align}  \label{e13-1} 
	{\left(\begin{array}{llll}
			1 & 1 & \cdots &  1\\
			\theta^{g_0} & \theta^{g_1} & \cdots &  \theta^{g_{b-1}}\\
			\ \ \vdots &  \ \ \vdots & \ddots &  \ \  \vdots\\
			\theta^{g_0(b-1)} & \theta^{g_1(b-1)} & \cdots &  \theta^{g_{b-1}(b-1)}\\
		\end{array}
		\right)} {\left(\begin{array}{l}  c_{i_pu+g_0,a}^{(y)} \\
			c_{i_pu+g_1,a}^{(y)} \\ \ \ \ \ \ \vdots \\
			c_{i_pu+g_{b-1},a}^{(y)}
		\end{array}
		\right)}
	={\left(\begin{array}{llll}
			\Delta_{0,a}^{(y)}-\sum\limits_{g\in{\cal J}}c_{i_pu+g,a}^{(y)} \\ \Delta_{1,a}^{(y)}-\sum\limits_{g\in{\cal J}}\theta^{g}c_{i_pu+g,a}^{(y)}
			\\ \ \ \indent \indent  \ \vdots\\ \Delta_{b-1,a}^{(y)}-\sum\limits_{g\in{\cal J}}\theta^{(b-1)g}c_{i_pu+g,a}^{(y)} \\
		\end{array}
		\right)},
\end{align}
where $a\in[\bar{s}^{\bar{n}}]$,
$y\in[\bar{s}+\bar{h}-\delta]$. Since $\theta^{g_i}\neq
\theta^{g_{i'}}$ for $g_i,g_{i'}\in[u]$ with $g_i\neq g_{i'}$, then the
Vandermonde matrix located on the left side of the linear system (\ref{e13-1})
is invertible, we can recover $\{c_{i_pu+g_0,a}^{(y)},
c_{i_pu+g_1,a}^{(y)}, \ldots, c_{i_pu+g_{b-1},a}^{(y)}:
a\in[\bar{s}^{\bar{n}}], y\in[\bar{s}+\bar{h}-\delta]\}$, i.e.,
the $b$ failed nodes in rack $i_p$.

In the following, we will
discuss the optimality of its repair bandwidth in two cases
depending on the value of $b$.

Case 1: $b\in[1,u-v]$.  We show that the code achieves the optimal repair
bandwidth for this case. In the download phase, by conclusion 1) of Lemma \ref{lm1}, each host rack
$i_p$ downloads $H_{i_p,j}(a,m)$, $a\in[\bar{s}^{\bar{n}}]$, $m\in[b]$ from all the helper racks $j\in {\cal R}$, i.e., $b\cdot\bar{s}^{\bar{n}}= bl/(\bar{s}+\bar{h}-\delta)$
symbols from each of the $\bar{d}$ helper racks. Consequently, the repair
bandwidth during this phase is
\begin{align}\label{e19}
	\displaystyle\gamma_{11}=\bar{d}\cdot\bar{h}\cdot\frac{bl}{\bar{s}+\bar{h}-\delta}=\frac{\bar{d}hl}{\bar{s}+\bar{h}-\delta}.
\end{align}
In the subsequent partially cooperative phase, each host rack
$i_p$ downloads $H_{i_q,i_p}(a,m)$, $a\in[\bar{s}^{\bar{n}}]$,
$m\in[b]$ from each host rack $i_q\in \mathcal{S}_p$,
respectively, i.e., $bl/(\bar{s}+\bar{h}-\delta)$ symbols from each of other
$\bar{h}-\delta$ host racks. We emphasize that these data are accessible in accordance with Lemma \ref{lm2}, which implies that the repair bandwidth during this phase is
\begin{align}\label{e20}
	\displaystyle\gamma_{12}=\bar{h}(\bar{h}-\delta)\cdot\frac{bl}{\bar{s}+\bar{h}-\delta}=
	\frac{h(\bar{h}-\delta)l}{\bar{s}+\bar{h}-\delta}.
\end{align}
From \eqref{e19} and \eqref{e20}, the total repair bandwidth across the two repair phases is
\begin{align*}
	\displaystyle\gamma_1=\gamma_{11}+\gamma_{12}=\frac{h(\bar{d}+\bar{h}-\delta)l}{\bar{s}+\bar{h}-\delta},
\end{align*}
\noindent which matches the minimum repair bandwidth given in \eqref{e2}.

Case 2: $b\in[u-v+1,u]$. We show that the repair bandwidth of the code $\mathcal{C}$ is asymptotically optimal when $\bar{d}$ is large
enough. In the download phase, by conclusion 1) of Lemma \ref{lm1}
each host rack $i_p$ downloads $H_{i_p,j}(a,m)$, $a\in[0,\bar{s}^{\bar{n}})$,  $m\in[0,u-v)$ from all the helper
racks $j\in {\cal R}$, i.e., $(u-v)l/(\bar{s}+\bar{h}-\delta)$ symbols
from each of the $\bar{d}$ helper racks. For each integer
$m\in[u-v,b)$, by conclusion 2) of Lemma \ref{lm1} each host rack $i_p$ downloads $H_{i_p,j}(a,m)$, $a\in[0,\bar{s}^{\bar{n}})$ from
all the helper racks $j\in \mathcal{R}\cup\{j'\}$ where $j'\in[\bar{n}]\setminus(\mathcal{R}\cup\mathcal{F})$, i.e.,
$(b-u+v)l/(\bar{s}+\bar{h}-\delta)$ symbols from each of the
$\bar{d}+1$ helper racks. Consequently, the repair bandwidth during this
phase is
\begin{align}\label{e20-1}
	\nonumber\displaystyle\gamma_{21}&=\bar{d}\cdot\bar{h}\cdot\frac{(u-v)l}{\bar{s}+\bar{h}-\delta}+(\bar{d}+1)\cdot\bar{h}\cdot\frac{(b-u+v)l}{\bar{s}+\bar{h}-\delta}\\
	&=\displaystyle\frac{\bar{d}hl}{\bar{s}+\bar{h}-\delta}+\frac{\bar{h}(b-u+v)l}{\bar{s}+\bar{h}-\delta}.
\end{align}
Exactly the same as discussed for Case 1, the repair bandwidth at
the  cooperative phase is
\begin{align}\label{e20-3}
	\displaystyle\gamma_{22}=\bar{h}(\bar{h}-\delta)\cdot\frac{bl}{\bar{s}+\bar{h}-\delta}=
	\frac{h(\bar{h}-\delta)l}{\bar{s}+\bar{h}-\delta}.
\end{align}
By \eqref{e20-1} and \eqref{e20-3}, the total repair bandwidth across the two repair phases is
\begin{align*}
	\displaystyle\gamma_2=\gamma_{21}+\gamma_{22}
	=\frac{h(\bar{d}+\bar{h}-\delta)l}{\bar{d}-\bar{k}+\bar{h}-\delta+1}
	+\frac{\bar{h}(b-u+v)l}{\bar{s}+\bar{h}-\delta}.
\end{align*}
\noindent Note that $v<u$, we have $b-u+v<b$, then
\begin{center}
	$\displaystyle\gamma_2=\frac{h(\bar{d}+\bar{h}-\delta)l}{\bar{s}+\bar{h}-\delta}
	+\frac{\bar{h}(b-u+v)l}{\bar{s}+\bar{h}-\delta}<\frac{h(\bar{d}+1+\bar{h}-\delta)l} {\bar{s}+\bar{h}-\delta}$.
\end{center}
In this case, the ratio of the number of downloaded symbols to the
optimal repair bandwidth presented in (\ref{e2}) is smaller than
$1+1/(\bar{d}+\bar{h}-\delta)$. As a result, the repair bandwidth of the code
approaches optimal level provided that the amount of helper racks $\bar{d}$ is sufficiently large when $b>u-v$. \qed

\subsection{Example 1}
\label{subsect-scheme-1} 

In this subsection, we provide a specific example of rack-aware MSCR codes to further illustrate the construction idea.  Now let us take the  $(n=12,k=5,l=3\times 2^6)$ rack-aware MSPCR code to introduce our method as follows. 

\subsubsection{Construct an MDS array code $\mathcal{C}$ with $\bar{l}=2^6$}

Let $\mathbb{F}$ be a finite field of size $25$. Let $(n,u,k,r,d,h,b,v,\delta)=(12,2,5,7,7,3,1,1,2)$, then we have
$\bar{n}=n/u=6$, $\bar{k}=(k-v)/u=2$, $\bar{r}=(r+v)/u=4$,
$\bar{d}=(d-u+b)/u=3$, $\bar{s}=\bar{d}-\bar{k}+1=2$, $\bar{h}=h/b=3$ and	$l=(\bar{s}+\bar{h}-\delta)\bar{s}^{\bar{n}}=3\times2^{6}$. Let $\xi$ be a primitive element of
$\mathbb{F}$ and $\theta$ be an element in $\mathbb{F}$ with multiplicative order
$u=2$. Consider an $(n=12,k=5,\bar{l}=2^6,d=7,h=3,\delta=2)$ code $\mathcal{C}=({\bf c}_{0},{\bf c}_{1},\ldots,{\bf c}_{11})$ defined by the following parity check equations over $\mathbb{F}$:
\begin{equation}\label{e13-5}
	\sum\limits_{i=0}^{5}\sum\limits_{g=0}^{1}\theta^{gt}\lambda_{i,a_i}^tc_{2i+g,a}=0, 
\end{equation}
\noindent where $t\in[0,7)$, $a\in[2^{6}]$, ${\bf c}_{2i+g}=(c_{2i+g,0},$ $c_{2i+g,1},\ldots,c_{2i+g,2^6-1})^\top$ is a column vector over $\mathbb{F}$ for $i\in[6]$ and $g=0,1$, 
$\lambda_{i,a_i}=\xi^{2i+a_i}$ for $i\in[6]$, $a_i=0,1$ is the $i$-th term of the binary expansion of $a=(a_{5},\ldots,a_1,a_0)$.

For all  $a\in[2^{6}]$, the equations in \eqref{e13-5} can be rewritten in matrix form as 
	\begin{align} \label{e15-2} 
		{\left(\begin{array}{lllllll}
				1 & 1 & 1 & 1 & \ldots & 1 & 1 \\
				\lambda_{0,a_0} & \theta\lambda_{0,a_0} & \lambda_{1,a_1} & \theta\lambda_{1,a_1} & \ldots &\lambda_{5,a_5} & \theta\lambda_{5,a_5} \\
				\vdots & \vdots & \vdots & \vdots & \ldots & \vdots & \vdots \\
				\lambda_{0,a_0}^6 & (\theta\lambda_{0,a_0})^6 & \lambda_{1,a_1}^6 & (\theta\lambda_{1,a_1})^6 & \ldots &\lambda_{5,a_5}^6 & (\theta\lambda_{5,a_5})^6
			\end{array}
			\right)} {\left(\begin{array}{l} 
				c_{0,a} \\
				c_{1,a} \\ 
				\vdots\\
				c_{11,a}
			\end{array}
			\right)}=0.
	\end{align}    
Note that $\lambda_{i,a_i}=\xi^{2i+a_i}$ and $\theta^2=1$, it is easy to see that 
$\lambda_{0,a_0},\theta\lambda_{0,a_0},\ldots,\lambda_{5,a_5},\theta\lambda_{5,a_5}$ are distinct,  each of the $r=7$ columns of the parity check matrix in \eqref{e15-2} has a rank of 7, so any $k=5$ out of $n=12$ elements in the set $\{c_{0,a},c_{1,a},\ldots,
c_{11,a}\}$ can recover the whole set. Since this holds for any $a\in[2^6]$,
we know that any 5 nodes of the code can recover the whole code, i.e., the code $\mathcal{C}$ has MDS property.

\subsubsection{ Construct an MDS array code $\widetilde{\mathcal{C}}$ with $l=3\times2^6$}

By generating 3 instances of the MDS array code $\mathcal{C}$,
we obtain our rack-aware MSPCR code $\widetilde{\mathcal{C}}=(\widetilde{\bf c}_{0},\widetilde{\bf c}_{1},\ldots,\widetilde{\bf c}_{11})$ defined by the following parity check equations over $\mathbb{F}$:
\begin{equation}\label{e13-6}
	\sum\limits_{i=0}^{5}\sum\limits_{g=0}^{1}\theta^{gt}\lambda_{i,a_i}^tc_{2i+g,a}^{(y)}=0, 
\end{equation}
\noindent where  $t\in[0,7)$, $a\in[2^{6}]$, $y\in[3]$, $\widetilde{\bf c}_{2i+g}=(({\bf c}_{2i+g}^{(0)})^\top,({\bf c}_{2i+g}^{(1)})^\top,({\bf c}_{2i+g}^{(2)})^\top)^\top$ is a column vector of length $l=3\times2^6$ over $\mathbb{F}$ for $i\in[0,6)$ and $g=0,1$.   Let the index set of
the $\bar{h}$ host racks be ${\cal F}=\{0,1,2\}$, the index set of the $\bar{d}$ helper racks be
${\cal R}=\{3,4,5\}$, and the index set of failed nodes be
${\cal F}'=\{0,2,4\}$. Since $t\in[0,7)$, taking $t=2\bar{t}+m$, we have $\bar{t}\in\{0,1,2,3\}$ when $m=0$. Substituting $t=2\bar{t}$ into the equation defined in \eqref{e13-6}, we have
\begin{equation}\label{e13-3}
	\sum\limits_{i=0}^{5}\lambda_{i,a_i}^{2\bar{t}}\sum\limits_{g=0}^{1}c_{2i+g,a}^{(y)}=0
\end{equation}
\noindent for any $a\in[2^6]$,  $\bar{t}\in\{0,1,2,3\}$ and $y\in[3]$.

\subsubsection{ Each host rack downloads some symbols from each of the helper racks}

For any host rack $p\in{\cal F}$,  $a\in[2^6]$ and $\bar{t}\in\{0,1,2,3\}$, by equation \eqref{e13-3} we can get
\begin{align*}
	\lambda_{p,a_p}^{2\bar{t}}\sum\limits_{g=0}^{1}c_{2p+g,a}^{(p)}+\sum\limits_{i\in{\cal F}\setminus\{p\}}\lambda_{i,a_i}^{2\bar{t}}\sum\limits_{g=0}^{1}c_{2i+g,a}^{(p)} +\sum\limits_{j\in{\cal R}}\lambda_{j,a_j}^{2\bar{t}}\sum\limits_{g=0}^{1}c_{2j+g,a}^{(p)}=0,
\end{align*}
\noindent and
\begin{align*}
	\lambda_{p,a_p\oplus1}^{2\bar{t}}\sum\limits_{g=0}^{1}c_{2p+g,a(p,a_p\oplus1)}^{(p+1)}+\sum\limits_{i\in{\cal F}\setminus\{p\}}\lambda_{i,a_i}^{2\bar{t}}\sum\limits_{g=0}^{1}c_{2i+g,a(p,a_p\oplus1)}^{(p+1)} +\sum\limits_{j\in{\cal R}}\lambda_{j,a_j}^{2\bar{t}}\sum\limits_{g=0}^{1}c_{2j+g,a(p,a_p\oplus1)}^{(p+1)}=0,
\end{align*}
\noindent which gives 
\begin{align}\label{e14}
	\nonumber	\lambda_{p,a_p}^{2\bar{t}}\sum\limits_{g=0}^{1}c_{2p+g,a}^{(p)}+\lambda_{p,a_p\oplus1}^{2\bar{t}}\sum\limits_{g=0}^{1}c_{2p+g,a(p,a_p\oplus1)}^{(p+1)}
	\nonumber +\sum\limits_{i\in{\cal F}\setminus\{p\}}\lambda_{i,a_i}^{2\bar{t}}(\sum\limits_{g=0}^{1}c_{2i+g,a}^{(p)}+\sum\limits_{g=0}^{1}c_{2i+g,a(p,a_p\oplus1)}^{(p+1)})\\
	=-\sum\limits_{j\in{\cal R}}\lambda_{j,a_j}^{2\bar{t}}(\sum\limits_{g=0}^{1}c_{2j+g,a}^{(p)}+\sum\limits_{g=0}^{1}c_{2j+g,a(p,a_p\oplus1)}^{(p+1)}),
\end{align}	
\noindent where the superscript operation is addition modulo $3$. Host rack $p$ downloads the data $\sum\limits_{g=0}^{1}c_{2j+g,a}^{(p)}+\sum\limits_{g=0}^{1}c_{2j+g,a(p,a_p\oplus1)}^{(p+1)}$, $a\in[2^6]$ from all the helper racks $j\in{\cal R}$, so the data on the right-hand side of equation (\ref{e14}) is known. We can obtain the coefficient matrix on the left-hand side of equation (\ref{e14}) as follow:

${\bf M}={\left(\begin{array}{llll}
		1 & 1 & 1 &  1 \\
		\lambda_{p,a_p}^2 & \lambda_{p,a_p\oplus1}^2 & \lambda_{i_0,a_{i_0}}^2 &  \lambda_{i_1,a_{i_1}}^2 \\
		\lambda_{p,a_p}^4 & \lambda_{p,a_p\oplus1}^4 & \lambda_{i_0,a_{i_0}}^4 &  \lambda_{i_1,a_{i_1}}^4 \\
		\lambda_{p,a_p}^6 & \lambda_{p,a_p\oplus1}^6 & \lambda_{i_0,a_{i_0}}^6 &  \lambda_{i_1,a_{i_1}}^6 \\ \\
	\end{array}
	\right)}={\left(\begin{array}{llll}
		1 & 1 & 1 &  1 \\
		\xi^{4p+2a_p} & \xi^{4p+2(a_p\oplus1)} & \xi^{4i_0+2a_{i_0}} &  \xi^{4i_1+2a_{i_1}} \\
		(\xi^{4p+2a_p})^2 & (\xi^{4p+2(a_p\oplus1)})^2 & (\xi^{4i_0+2a_{i_0}})^2 &  (\xi^{4i_1+2a_{i_1}})^2 \\
		(\xi^{4p+2a_p})^3 & (\xi^{4p+2(a_p\oplus1)})^3 & (\xi^{4i_0+2a_{i_0}})^3 &  (\xi^{4i_1+2a_{i_1}})^3 \\
	\end{array}
	\right)},$\\

\noindent where ${\cal F}\setminus\{p\}=\{i_0,i_1\}$.

Clearly, $\xi^{4p+2a_p}, \xi^{4p+2(a_p\oplus1)}, \xi^{4i_0+2a_{i_0}}, \xi^{4i_1+2a_{i_1}}$ are four distinct elements in field $\mathbb{F}$,  the rank of the Vandermonde matrix ${\bf M}$ is 4. Then, host rack $p$ can recover symbols $\sum\limits_{g=0}^{1}c_{2p+g,a}^{(p)}$, $\sum\limits_{g=0}^{1}c_{2p+g,a}^{(p+1)}$, $\sum\limits_{g=0}^{1}c_{2i_0+g,a}^{(p)}+\sum\limits_{g=0}^{1}c_{2i_0+g,a(p,a_p\oplus1)}^{(p+1)}$, $\sum\limits_{g=0}^{1}c_{2i_1+g,a}^{(p)}+\sum\limits_{g=0}^{1}c_{2i_1+g,a(p,a_p\oplus1)}^{(p+1)}$, $a\in[2^6]$.

\subsubsection{ Cooperatively repair the failed nodes}

For host rack $0$, host rack $1$ transfers $\sum\limits_{g=0}^{1}c_{g,a}^{(1)}+\sum\limits_{g=0}^{1}c_{g,a(1,a_1\oplus1)}^{(2)}$, $a\in[2^6]$
to host rack $0$. Host rack $0$ uses its own data $\sum\limits_{g=0}^{1}c_{g,a}^{(1)}$, $a\in[2^6]$ to solve 
$\sum\limits_{g=0}^{1}c_{g,a(1,a_1\oplus1)}^{(2)}$, i.e., $\sum\limits_{g=0}^{1}c_{g,a}^{(2)}$, $a\in[2^6]$. So, host rack $0$ recovers  $\sum\limits_{g=0}^{1}c_{g,a}^{(0)}$, $\sum\limits_{g=0}^{1}c_{g,a}^{(1)}$, $\sum\limits_{g=0}^{1}c_{g,a}^{(2)}$, $a\in[2^6]$.
Using  $c_{1,a}^{(p)}, a\in[2^{6}],p=0,1,2$ of helper node
$1$ in host rack $0$, we can solve $c_{0,a}^{(p)}, a\in[2^{6}], p=0,1,2$ stored in failed node $0$.

For host rack $1$,  host rack $2$ transfers  $\sum\limits_{g=0}^{1}c_{2+g,a}^{(2)}+\sum\limits_{g=0}^{1}c_{2+g,a(2,a_2\oplus1)}^{(0)}$, $a\in[2^6]$ to host rack $1$. For host rack $2$, host rack $1$ transfers  $\sum\limits_{g=0}^{1}c_{4+g,a}^{(1)}+\sum\limits_{g=0}^{1}c_{4+g,a(1,a_1\oplus1)}^{(2)}$, $a\in[2^6]$ to host rack $2$. Using the same method as that for repairing failed node 0, we can repair failed nodes $2$ and $4$, respectively.

In the download phase, we downloaded $2^6\times3\times3=576$  symbols from the 3 helper racks, while we downloaded $2^6\times3=192$ in the cooperative phase. So, in the repair scheme we downloaded $768$ symbols, which achieves the lower bound in \eqref{e2}. A numerical comparisons with existing schemes is provided in the following table.

\begin{table}[http!]
	\renewcommand{\arraystretch}{1.2}
	\setlength\tabcolsep{3pt} 
	\centering
	\caption{ The Existing Schemes And Our Schemes For $(n,u,k,d,h,b,v)=(12,2,5,7,3,1,1)$
		\label{tab2}
	} 
	\begin{tabular}{|c|c|c|c|c|c|c|c|c|}	
		\hline
		&  Sub-packetization  $l$ & \tabincell{c}{Optimal repair\\ bandwidth $\gamma$} &  \tabincell{c}{ Field  size $|\mathbb{F}|$} \\  
		\hline
		\cite{GDL} & $4\times2^6$ & $15\times 2^6$ & 25\\ 
		\hline
		Our schemes &	 $3\times2^6$ & $12\times 2^6$& 25\\ 		
		\hline
	\end{tabular}
\end{table}
It is not difficult to see that the sub-packetization and optimal repair bandwidth of our scheme are significantly smaller than those in \cite{GDL}, while the fields employed in both schemes are of the same size.

\section{The proof of Theorem \ref{th3}}
\label{section-th3}

Since $\bar{d}=\bar{k}+1$, then $\bar{s}=\bar{d}-\bar{k}+1=2$. Let $\mathbb{F}$ be a finite field of size
$|\mathbb{F}|\geq 2n+1$ with $u|(|\mathbb{F}|-1)$, $\xi$ be a primitive element of $\mathbb{F}$, $\theta$ be an
element of $\mathbb{F}$ with multiplicative order $u$,
$\lambda_{i,j}=\xi^{2i+j}$ for $i\in[\bar{n}]$,
$j\in\{0,1\}$ be $2\bar{n}$ distinct elements over $\mathbb{F}$, and $l=2^{\bar{n}}$.
From the rack-aware storage model, we know that the number of helper nodes is $d=\bar{d}u+u-b$. Similarly, to guarantee $d\geq k$, we need $\bar{k}\leq\bar{d}\leq \bar{n}-\bar{h}$ when $1\leq b\leq u-v$, and $\bar{k}+1\leq\bar{d}\leq \bar{n}-\bar{h}$ when $u-v<b\leq u$. Consider a code $\mathcal{C}=({\bf c}_0,{\bf c}_1,\ldots,{\bf
	c}_{n-1})$ defined by the following parity check equations over $\mathbb{F}$:
\begin{equation}\label{e21}
	\sum\limits_{i=0}^{\bar{n}-1}\sum\limits_{g=0}^{u-1}\theta^{gt}\lambda_{i,a_i}^tc_{iu+g,a}=0,
\end{equation}
\noindent where $t\in[r]$, $a\in[l]$, ${\bf c}_{iu+g}=(c_{iu+g,0},c_{iu+g,1},\ldots,c_{iu+g,l-1})^\top$ is a column vector of length $l$ over $\mathbb{F}$.

Before showing the proof of  Theorem \ref{th3}, we first review the Hamming code, which plays an important role in this repair scheme.
Let $m'\geq2$ be an integer. Given an $(n'=2^{m'}-1, k'=2^{m'}-1-m')$ linear binary
Hamming code $\mathcal{V}_0$ with parity check matrix whose columns contains all the nonzeros vectors of $\mathbb{F}_2^{m'}$ where $\mathbb{F}_2$ is a binary field. For each $i\in[n']$ we define $n'$ sets
\begin{equation}\label{4-e1}
	\mathcal{V}_{i+1}=\{{\bf f}(i,f_i\oplus1):{\bf f}=(f_0,\ldots,f_{n'-1})\in \mathcal{V}_0\},
\end{equation}
\noindent where $\oplus$ is the addition operation modulo $2$.

\begin{lemma}[\cite{LWCT}]\label{lm4-1}
	The sets $\mathcal{V}_{i+1}$ where $i\in[n']$ defined in (\ref{4-e1}) and $\mathcal{V}_0$ form a partition of $[2^{n'}]$, i.e., $\mathcal{V}_0\cup \mathcal{V}_1\cup\cdots\cup \mathcal{V}_{n'}=[2^{n'}]$. 
\end{lemma}

Note that $l=2^{\bar{n}}$, then $\bar{h}-\delta+2$ implies that there exists
an integer $1\leq \bar{m}\leq \bar{n}$ such that $\bar{h}-\delta+2=2^{\bar{m}}$. For convenience, let $\mathcal{V}_0$ be a $(2^{\bar{m}}-1,2^{\bar{m}}-1-\bar{m})$  linear binary Hamming code. For any
index set of the host racks ${\mathcal F}=\{i_0,i_1,\ldots,i_{\bar{h}-1}\}$,
define a partition $\{\mathcal{S}_q:q\in[0,\bar{h}/(\bar{h}-\delta+1))\}$ of ${\mathcal F}$,
where $\bar{h}$ is divisible by $\bar{h}-\delta+1$ and
\begin{align}\label{lm4-0}
	\mathcal{S}_q=\{i_p:p\in[q(\bar{h}-\delta+1),(q+1)(\bar{h}-\delta+1))\}, q\in[0,\bar{h}/(\bar{h}-\delta+1)).
\end{align} 
Our cooperative repair scheme relies on $\mathcal{S}_q$ defined above. Let $J=\{j_0,j_1,\ldots, j_{i-1}\}$ be a subset of $[\bar{n}]$. Given any element
$a=(a_{\bar{n}-1},a_{\bar{n}-2},\ldots,a_0)$ belonging to $ \mathbb{F}_2^{\bar{n}}$, we define the puncturing
vector over $J$ as $a|_J=(a_{j_{i-1}},a_{j_{i-2}},\ldots,a_{j_0})\in \mathbb{F}_2^{i}$,
under the assumption that $j_0<j_1<\ldots<j_{i-1}$.

In this section, for any two distinct racks $i$ and $j$, we define 
\begin{center}
	$H_{i,j}(a,m)=\sum\limits_{g=0}^{u-1}\theta^{gm}c_{ju+g,a}+\sum\limits_{g=0}^{u-1}\theta^{gm}c_{ju+g,a(i,a_{i}\oplus1)}$,
\end{center}
\noindent where $m\in[u]$ and $a\in[l]$.

When there are $\bar{h}$ host racks and each contains $b$ failed nodes, the partially cooperative repair process of the code generated by \eqref{e21}  is mainly split into two phases.

$\bullet$ Download phase: When $b\in[1,u-v]$,  each host rack $i_p\in \mathcal{S}_{\lfloor p/(\bar{h}-\delta+1)\rfloor}$ downloads $H_{i_p,j}(a,m)$ where $m\in [b]$, $a\in[l]$ and $ a|_{\mathcal{S}_{\lfloor p/(\bar{h}-\delta+1)\rfloor}}\in {\mathcal V}_0$,  from each helper rack $j\in{\mathcal R}$. When $b\in[u-v+1,u]$, each host rack $i_p\in \mathcal{S}_{\lfloor p/(\bar{h}-\delta+1)\rfloor}$ downloads $H_{i_p,j}(a,m)$ where $m\in [u-v]$, $a\in[l]$ and $ a|_{\mathcal{S}_{\lfloor p/(\bar{h}-\delta+1)\rfloor}}\in {\mathcal V}_0$,  from each helper rack $j\in{\mathcal R}$, and additionally downloads $H_{i_p,j}(a,m)$ where $m\in [u-v,b)$, $a\in[l]$ and $ a|_{\mathcal{S}_{\lfloor p/(\bar{h}-\delta+1)\rfloor}}\in {\mathcal V}_0$,  from each helper rack $j\in{\mathcal R}\cup\{j'\}$, where $j'\in[\bar{n}]\setminus(\mathcal{R}\cup\mathcal{F})$.

$\bullet$ Cooperative phase: each host rack $i\in
\mathcal{S}_{\lfloor p/(\bar{h}-\delta+1)\rfloor}\setminus\{i_p\}$ sends $H_{i,i_p}(a,m)$ where $m\in [b]$, $a\in[l]$ and $ a|_{\mathcal{S}_{\lfloor p/(\bar{h}-\delta+1)\rfloor}}\in {\mathcal V}_0$, to host rack $i_p$.

In the following, we further show that the repair method can recover $h$ failures distributed evenly in $\bar{h}$ host racks with the (asymptotially) optimal repair bandwidth. During the download phase, we give the following result about the repair process.
\begin{lemma}\label{lm4-2}
	According to the notations introduced above, we can obtain the following results:

1) For any given integer $b\in[1,u-v]$, by downloading $H_{i_p,j}(a,m)$, $m\in[b]$ from every helper rack $j\in{\cal R}$, host rack $i_p$ $(p\in[\bar{h}])$ can recover

$\bullet$ $\sum\limits_{g=0}^{u-1}\theta^{gm}c_{i_pu+g,a}$,
$\sum\limits_{g=0}^{u-1}\theta^{gm}c_{i_pu+g,a(i_p,a_{i_p}\oplus1)}$;

$\bullet$
$H_{i_p,i}(a,m)=\sum\limits_{g=0}^{u-1}\theta^{gm}c_{iu+g,a}+\sum\limits_{g=0}^{u-1}\theta^{gm}c_{iu+g,a(i_p,a_{i_p}\oplus1)}$,

\noindent where $i\in{\cal F}\setminus\{i_p\}$, $a\in[l]$ with $a|_{\mathcal{S}_{\lfloor
		p/(\bar{h}-\delta+1)\rfloor}}\in {\mathcal V}_0$.

2) For any given integer $b\in[u-v+1,u]$, by downloading $H_{i_p,j}(a,m)$,  $m\in[u-v,b)$ from every helper rack $j\in{\cal R}\cup\{j'\}$ with
$j'\in[\bar{n}]\setminus(\mathcal{R}\cup\mathcal{F})$,  host rack $i_p$
$(p\in[\bar{h}])$ can recover

$\bullet$ $\sum\limits_{g=0}^{u-1}\theta^{gm}c_{i_pu+g,a}$,
$\sum\limits_{g=0}^{u-1}\theta^{gm}c_{i_pu+g,a(i_p,a_{i_p}\oplus1)}$;

$\bullet$ $H_{i_p,i}(a,m)$, 

\noindent where $i\in{\cal F}\setminus\{i_p\}$, $a\in[l]$ with $a|_{\mathcal{S}_{\lfloor
		p/(\bar{h}-\delta+1)\rfloor}}\in {\mathcal V}_0$.
\end{lemma}
\begin{proof}
	We only prove conclusion 1),  conclusion 2) can
	be proved similarly. Let $t=\bar{t}u+m\in[r]$. Since $m\in[b]$ and $b\in[1,u-v]$, we have $m\in[0,u-v)$, then $\bar{t}\in[\bar{r}]$. Since $\theta^u=1$, from
	the code $\mathcal{C}$ defined in \eqref{e21}  we can get
	\begin{align*}
		\sum\limits_{i=0}^{\bar{n}-1}\lambda_{i,a_i}^{\bar{t}u+m}\sum\limits_{g=0}^{u-1}\theta^{gm}c_{iu+g,a}=0,
	\end{align*}
	where $a\in[l]$, $\bar{t}\in[\bar{r}]$ and $m\in[b]$, which implies
	\begin{align*}
		&\lambda_{i_p,a_{i_p}}^{\bar{t}u+m}\sum\limits_{g=0}^{u-1}\theta^{gm}c_{i_pu+g,a}+
		\sum\limits_{i\in{\cal F}\setminus\{i_p\}}\lambda_{i,a_{i}}^{\bar{t}u+m}\sum\limits_{g=0}^{u-1}\theta^{gm}c_{iu+g,a}\\
		&+\sum\limits_{j\in{\cal R}}\lambda_{j,a_{j}}^{\bar{t}u+m}\sum\limits_{g=0}^{u-1}\theta^{gm}c_{ju+g,a
		}+\sum\limits_{z\in{\cal U}}\lambda_{z,a_{z}}^{\bar{t}u+m}\sum\limits_{g=0}^{u-1}\theta^{gm}c_{zu+g,a}=0,
	\end{align*}
	\noindent and
	\begin{align*}
		&\lambda_{i_p,a_{i_p}\oplus
			1}^{\bar{t}u+m}\sum\limits_{g=0}^{u-1}\theta^{gm}c_{i_pu+g,a(i_p,a_{i_p}\oplus
			1)}+\sum\limits_{i\in{\cal F}\setminus\{i_p\}}\lambda_{i,a_{i}}^{\bar{t}u+m}\sum\limits_{g=0}^{u-1}\theta^{gm}c_{iu+g,a(i_p,a_{i_p}\oplus
			1)}\\
		&+\sum\limits_{j\in{\cal R}}\lambda_{j,a_{j}}^{\bar{t}u+m}\sum\limits_{g=0}^{u-1}\theta^{gm}c_{ju+g,a(i_p,a_{i_p}\oplus
			1)}+\sum\limits_{z\in{\cal U}}\lambda_{z,a_{z}}^{\bar{t}u+m}\sum\limits_{g=0}^{u-1}\theta^{gm}c_{zu+g,a(i_p,a_{i_p}\oplus
			1)}=0.
	\end{align*}

	\noindent Summing the two equations above yields
	
	$\displaystyle\lambda_{i_p,a_{i_p}}^{\bar{t}u+m}\sum\limits_{g=0}^{u-1}\theta^{gm}c_{i_pu+g,a}+\lambda_{i_p,a_{i_p}\oplus
		1}^{\bar{t}u+m}\sum\limits_{g=0}^{u-1}\theta^{gm}c_{i_pu+g,a(i_p,a_{i_p}\oplus1)}+\sum\limits_{i\in{\cal F}\setminus\{i_p\}}\lambda_{i,a_{i}}^{\bar{t}u+m}(\sum\limits_{g=0}^{u-1}\theta^{gm}c_{iu+g,a}
	+\sum\limits_{g=0}^{u-1}\theta^{gm}c_{iu+g,a(i_p,a_{i_p}\oplus 1)})$
	
	\begin{equation} \label{4-e2}
	+\sum\limits_{z\in{\cal U}}\lambda_{z,a_{z}}^{\bar{t}u+m}(\sum\limits_{g=0}^{u-1}\theta^{gm}c_{zu+g,a}+\sum\limits_{g=0}^{u-1}\theta^{gm}c_{zu+g,a(i_p,a_{i_p}\oplus
			1)})=-\sum\limits_{j\in{\cal R}}\lambda_{j,a_{j}}^{\bar{t}u+m}(\sum\limits_{g=0}^{u-1}\theta^{gm}c_{ju+g,a}+\sum\limits_{g=0}^{u-1}\theta^{gm}c_{ju+g,a(i_p,a_{i_p}\oplus
			1)}),
	\end{equation}
	where $\bar{t}\in[\bar{r}]$, $m\in[b]$, $a\in[l]$ with
	$a|_{\mathcal{S}_{\lfloor p/(\bar{h}-\delta+1)\rfloor}}\in {\mathcal V}_0$. Let
	\begin{align*}
		{\bf v}_{i_p,a,m}^\top=&[\sum\limits_{g=0}^{u-1}\theta^{gm}c_{i_pu+g,a},
		\sum\limits_{g=0}^{u-1}\theta^{gm}c_{i_pu+g,a(i_p,a_{i_p}\oplus1)},H_{i_p,i_0}(a,m),\ldots,
		H_{i_p,i_{p-1}}(a,m),\\
		&H_{i_p,i_{p+1}}(a,m),\ldots,H_{i_p,i_{\bar{h}-1}}(a,m),H_{i_p,z_0}(a,m),\ldots,H_{i_p,z_{\bar{n}-\bar{h}-\bar{d}-1}}(a,m)]\in
		\mathbb{F}^{\bar{r}},
	\end{align*}
	\noindent and
	\begin{center}
		${\bf u}_{i_p,a,m}^\top=[-H_{i_p,j_0}(a,m),-H_{i_p,j_1}(a,m),\ldots,-H_{i_p,j_{\bar{k}}}(a,m)]\in
		\mathbb{F}^{\bar{k}+1}$
	\end{center}
	\noindent for ${\cal R}=\{j_0,j_1,\ldots,j_{\bar{k}}\}$, then (\ref{4-e2})
	can be rewritten as
	\begin{center}
		${\bf M}_{i_p,a,m}{\bf v}_{i_p,a,m}={\bf M}'_{i_p,a,m}{\bf u}_{i_p,a,m}$,
	\end{center}
	\noindent where
	{\footnotesize 
		\[{\bf M}_{i_p,a,m}=
		\left(
		\begin{array}{ccccccc}
			\lambda_{i_p,a_{i_p}}^{m} \ \ \ \ \lambda_{i_p,a_{i_p}\oplus1}^{m} &
			\ \lambda_{i_0,a_{i_0}}^{m}  \ldots
			\lambda_{i_{p-1},a_{i_{p-1}}}^{m} \ &
			\lambda_{i_{p+1},a_{i_{p+1}}}^{m} \ldots
			\lambda_{i_{\bar{h}-1},a_{i_{\bar{h}-1}}}^{m} &
			\lambda_{z_0,a_{z_0}}^{m}\ldots \lambda_{z_{\bar{n}-\bar{h}-\bar{d}-1},a_{z_{\bar{n}-\bar{h}-\bar{d}-1}}}^{m}\\
			\lambda_{i_p,a_{i_p}}^{u+m} \ \ \ \lambda_{i_p,a_{i_p}\oplus1}^{u+m} &
			\ \ \lambda_{i_0,a_{i_0}}^{u+m} \ldots
			\lambda_{i_{p-1},a_{i_{p-1}}}^{u+m} &
			\lambda_{i_{p+1},a_{i_{p+1}}}^{u+m} \ldots
			\lambda_{i_{\bar{h}-1},a_{i_{\bar{h}-1}}}^{u+m} &
			\lambda_{z_0,a_{z_0}}^{u+m}\ldots \lambda_{z_{\bar{n}-\bar{h}-\bar{d}-1},a_{z_{\bar{n}-\bar{h}-\bar{d}-1}}}^{u+m}\\
			\ \ \vdots \indent \indent\indent \indent   \vdots  & \vdots \indent\indent \indent \indent  \vdots & \vdots
			\indent\indent\indent\indent\indent\indent \vdots & \vdots
			\indent\indent \indent\indent \indent \vdots \indent\indent \indent\indent\indent \indent \ \ \\
			\lambda_{i_p,a_{i_p}}^{(\bar{r}-1)u+m}
			\lambda_{i_p,a_{i_p}\oplus1}^{(\bar{r}-1)u+m} &
			\lambda_{i_0,a_{i_0}}^{(\bar{r}-1)u+m}  \ldots
			\lambda_{i_{p-1},a_{i_{p-1}}}^{(\bar{r}-1)u+m} &
			\lambda_{i_{p+1},a_{i_{p+1}}}^{(\bar{r}-1)u+m} \ldots
			\lambda_{i_{\bar{h}-1},a_{i_{\bar{h}-1}}}^{(\bar{r}-1)u+m} &
			\lambda_{z_0,a_{z_0}}^{(\bar{r}-1)u+m}\ldots \lambda_{z_{\bar{n}-\bar{h}-\bar{d}-1},a_{z_{\bar{n}-\bar{h}-\bar{d}-1}}}^{(\bar{r}-1)u+m}\\
		\end{array}
		\right),
		\]
	}
	\noindent and
	{\footnotesize
		\[ {\bf M}'_{i_p,a,m}=
		\left(
		\begin{array}{ccccccc}
			\lambda_{j_0,a_{j_0}}^{m} & \lambda_{j_1,a_{j_1}}^{m} & \cdots &
			\lambda_{j_{\bar{k}},a_{j_{\bar{k}}}}^{m}\\
			\lambda_{j_0,a_{j_0}}^{u+m} & \lambda_{j_1,a_{j_1}}^{u+m} & \cdots &
			\lambda_{j_{\bar{k}},a_{j_{\bar{k}}}}^{u+m}\\
			\vdots & \vdots & \ddots &
			\vdots\\
			\lambda_{j_0,a_{j_0}}^{(\bar{r}-1)u+m} &
			\lambda_{j_1,a_{j_1}}^{(\bar{r}-1)u+m} & \cdots &
			\lambda_{j_{\bar{k}},a_{j_{\bar{k}}}}^{(\bar{r}-1)u+m}\\
		\end{array}
		\right).
		\]
	}
	
	Note that  host rack $i_p$ downloads
	\begin{center}
		$H_{i_p,j}(a,m)=\sum\limits_{g=0}^{u-1}\theta^{gm}c_{ju+g,a}+\sum\limits_{g=0}^{u-1}\theta^{gm}c_{ju+g,a(i_p,a_{i_p}\oplus1)}$, $m\in[b]$, 
		$a\in[l]$ with $a|_{\mathcal{S}_{\lfloor
				p/(\bar{h}-\delta+1)\rfloor}}\in {\mathcal V}_0$
	\end{center}
	\noindent from all helper racks $j\in{\cal R}$, thus ${\bf u}_{i_p,a,m}$ is
	obtained. The matrix ${\bf M}_{i_p,a,m}$ can be rewritten as
	${\bf M}_1$diag$(\lambda_{i_p,a_{i_p}}^{m},\lambda_{i_p,a_{i_p}\oplus1}^{m},$
	$\ldots, \lambda_{z_0,a_{z_0}}^{m},\ldots,\lambda_{z_{\bar{n}-\bar{h}-\bar{d}-1},a_{z_{\bar{n}-\bar{h}-\bar{d}-1}}}^{m})$,
	where
	
	{\footnotesize 
		\[{\bf M}_1=
		\left(
		\begin{array}{ccccccc}
			1 \indent \indent \indent 1\indent\indent\indent &
			\ 1 \indent  \ldots\indent 	1 \indent\indent\indent \ &
			1\indent\indent  \ldots\indent \indent  1 \indent\indent\indent &	1\indent  \ldots \indent  1\indent\indent\indent\indent\indent\indent\indent\indent \\
			\lambda_{i_p,a_{i_p}}^{u} \ \ \ \lambda_{i_p,a_{i_p}\oplus1}^{u} &
			\ \ \lambda_{i_0,a_{i_0}}^{u} \ldots
			\lambda_{i_{p-1},a_{i_{p-1}}}^{u} &
			\lambda_{i_{p+1},a_{i_{p+1}}}^{u} \ldots
			\lambda_{i_{\bar{h}-1},a_{i_{\bar{h}-1}}}^{u} &
			\lambda_{z_0,a_{z_0}}^{u}\ldots \lambda_{z_{\bar{n}-\bar{h}-\bar{d}-1},a_{z_{\bar{n}-\bar{h}-\bar{d}-1}}}^{u}\\
			\ \ \vdots \indent \indent\indent    \vdots\indent  & \vdots \indent\indent \indent \indent  \vdots & \vdots
			\indent\indent\indent\indent\indent\indent \vdots & \vdots
			\indent\indent \indent\indent \indent \vdots \indent\indent \indent\indent\indent \indent \ \ \\
			\lambda_{i_p,a_{i_p}}^{(\bar{r}-1)u} \ \
			\lambda_{i_p,a_{i_p}\oplus1}^{(\bar{r}-1)u} &
			\lambda_{i_0,a_{i_0}}^{(\bar{r}-1)u}  \ldots
			\lambda_{i_{p-1},a_{i_{p-1}}}^{(\bar{r}-1)u} &
			\lambda_{i_{p+1},a_{i_{p+1}}}^{(\bar{r}-1)u} \ldots
			\lambda_{i_{\bar{h}-1},a_{i_{\bar{h}-1}}}^{(\bar{r}-1)u} &
			\lambda_{z_0,a_{z_0}}^{(\bar{r}-1)u+m}\ldots \lambda_{z_{\bar{n}-\bar{h}-\bar{d}-1},a_{z_{\bar{n}-\bar{h}-\bar{d}-1}}}^{(\bar{r}-1)u}\\
		\end{array}
		\right),
		\]
	}	
\begin{spacing}{1.1}	
Since $\lambda_{i,j}=\xi^{2i+j}$, then we have $\lambda_{i,j}^u\neq
\lambda_{i',j'}^u$ for all $i,i'\in[\bar{n}]$, $j,j'\in\{0,1\}$ with
$(i,j)\neq(i',j')$,  the Vandermonde matrix ${\bf M}_1$ are invertible
for each $a\in[l]$ with $a|_{\mathcal{S}_{\lfloor	p/(\bar{h}-\delta+1)\rfloor}}\in {\mathcal V}_0$. It follows that the matrix ${\bf M}_{i_p,a,m}$ is also invertible, therefore
\end{spacing}\begin{align*}
		{\bf v}_{i_p,a,m}={\bf M}_{i_p,a,m}^{-1}M'_{i_p,a,m}{\bf u}_{i_p,a,m}, m\in[b],
		a\in[l] \ {\rm with} \ a|_{\mathcal{S}_{\lfloor
				p/(\bar{h}-\delta+1)\rfloor}}\in {\mathcal V}_0.
\end{align*}
	\noindent In other words,  host rack $i_p$ can recover
	\begin{center}
		$\sum\limits_{g=0}^{u-1}\theta^{gm}c_{i_pu+g,a}$,
		$\sum\limits_{g=0}^{u-1}\theta^{gm}c_{i_pu+g,a(i_p,a_{i_p}\oplus1)}, m\in[b], a\in[l]$
		with $a|_{\mathcal{S}_{\lfloor p/(\bar{h}-\delta+1)\rfloor}}\in {\mathcal V}_0$,
	\end{center}
	\noindent and
	\begin{center}
		$H_{i_p,i}(a,m)=\sum\limits_{g=0}^{u-1}\theta^{gm}c_{iu+g,a}+\sum\limits_{g=0}^{u-1}\theta^{gm}c_{iu+g,a(i_p,a_{i_p}\oplus1)}$,
		$a\in[l]$ with $a|_{\mathcal{S}_{\lfloor
				p/(\bar{h}-\delta+1)\rfloor}}\in {\mathcal V}_0$
	\end{center}
	\noindent for any $i\in{\cal F}\setminus\{i_p\}$. 
\end{proof}

In regard to the partially cooperative phase, we have the following result.

\begin{lemma} \label{lm4-3}
	For any host rack $p\in [0,\bar{h})$ and any given integer $b\in[1,u]$, based
	on the data $\{\sum\limits_{g=0}^{u-1}\theta^{gm}c_{i_pu+g,a},$
	$\sum\limits_{g=0}^{u-1}\theta^{gm}c_{i_pu+g,a(i_p,a_{i_p}\oplus1)}
	: a\in[l],a|_{\mathcal{S}_{\lfloor p/(\bar{h}-\delta+1)\rfloor}}\in
	{\mathcal V}_0,m\in[b]\}$ from  lemma \ref{lm4-2}, some linear combinations of symbols in the $p$-th host
	rack
	$\{\sum\limits_{g=0}^{u-1}\theta^{gm}c_{i_pu+g,a}:a\in[l],m\in[b]\}$
	can be recovered if it downloads the data
	\begin{center}
		$H_{i,i_p}(a,m)=\sum\limits_{g=0}^{u-1}\theta^{gm}c_{i_pu+g,a}+\sum\limits_{g=0}^{u-1}\theta^{gm}c_{i_pu+g,a(i,a_{i}\oplus1)}$, 
		$a\in[l],a|_{\mathcal{S}_{\lfloor p/(\bar{h}-\delta+1)\rfloor}}\in
		{\mathcal V}_0$, $m\in[b]$
	\end{center}
	\noindent from $\bar{h}-\delta$ host racks $i\in \mathcal{S}_{\lfloor
		p/(\bar{h}-\delta+1)\rfloor}\setminus\{i_p\}$, where $\mathcal{S}_{\lfloor
		p/(\bar{h}-\delta+1)\rfloor}$ is defined in $(\ref{lm4-0})$.
\end{lemma}
\begin{proof}
	For each $p\in[0,\bar{h})$, the data
	$\{\sum\limits_{g=0}^{u-1}\theta^{gm}c_{i_pu+g,a(i,a_i\oplus1)}:a\in[l],a|_{\mathcal{S}_{\lfloor
			p/(\bar{h}-\delta+1)\rfloor}}\in {\mathcal V}_0, m\in[b]\}$ can be recovered by the
	known data
	\begin{center}
		$\{\sum\limits_{g=0}^{u-1}\theta^{gm}c_{i_pu+g,a}
		:a\in[l],a|_{\mathcal{S}_{\lfloor p/(\bar{h}-\delta+1)\rfloor}}\in
		{\mathcal V}_0, m\in[b]\}$,
	\end{center}
	\noindent and the data
	\begin{center}
		$\{\sum\limits_{g=0}^{u-1}\theta^{gm}c_{i_pu+g,a}+\sum\limits_{g=0}^{u-1}\theta^{gm}c_{i_pu+g,a(i,a_{i}\oplus1)}$,
		$a\in[l],a|_{\mathcal{S}_{\lfloor p/(\bar{h}-\delta+1)\rfloor}}\in
		{\mathcal V}_0, m\in[b]\}$
	\end{center}
	\noindent downloaded from host rack $i\in \mathcal{S}_{\lfloor
		p/(\bar{h}-\delta+1)\rfloor}\setminus\{i_p\}$. Thus, for any
	$p\in[0,\bar{h})$, host rack $i_p$ can obtain the
	data
	\begin{align*}
		&\{\sum\limits_{g=0}^{u-1}\theta^{gm}c_{i_pu+g,a},\sum\limits_{g=0}^{u-1}\theta^{gm}c_{i_pu+g,a(i_p,a_{i_p}\oplus1)}
		:a\in[l],a|_{\mathcal{S}_{\lfloor p/(\bar{h}-\delta+1)\rfloor}}\in
		{\mathcal V}_0,m\in[b]\}\\
		&\bigcup(\bigcup_{i\in \mathcal{S}_{\lfloor
				p/(\bar{h}-\delta+1)\rfloor}\setminus\{i_p\}}\{\sum\limits_{g=0}^{u-1}\theta^{gm}c_{i_pu+g,a(i,a_i\oplus1)}:a\in[l],a|_{\mathcal{S}_{\lfloor
				p/(\bar{h}-\delta+1)\rfloor}}\in {\mathcal V}_0,m\in[b]\})\\
		&=\{\sum\limits_{g=0}^{u-1}\theta^{gm}c_{i_pu+g,a}
		:a\in[l], m\in[b]\},
	\end{align*}
	\noindent where the equality holds
	from Lemma \ref{lm4-1}.
\end{proof}

We are now in a position to prove Theorem \ref{th3}.

{\bf The Proof of Theorem \ref{th3}: } We first recover the $b$ failed nodes
in racks $i_p(p=0,1,\ldots,\bar{h}-1)$. Assume that the index
set of the $b$ failed nodes in host rack $i_p$ is
${\cal I}=\{g_0,g_1,\ldots,g_{b-1}\}$ and ${\cal J}=[u]\setminus{\cal I}$. For any $a\in[l]$ and $m\in[b]$, by Lemma \ref{lm4-3} we know the data
$\Delta_{m,a}=\sum\limits_{g=0}^{u-1}\theta^{gm}c_{i_pu+g,a}$, thus
\begin{center}
	$\sum\limits_{g\in{\cal I}}\theta^{gm}c_{i_pu+g,a}=\Delta_{m,a}-\sum\limits_{g\in{\cal J}}\theta^{gm}c_{i_pu+g,a}$.
\end{center}
\noindent When $m$ runs over $\{0,1,\ldots,b-1\}$, we get
\noindent\noindent
\begin{equation} \label{xx1} {\left(\begin{array}{llll}
			1 & 1 & \cdots &  1\\
			\theta^{g_0} & \theta^{g_1} & \cdots &  \theta^{g_{b-1}}\\
			\ \ \vdots &  \ \ \vdots & \ddots &  \ \  \vdots\\
			\theta^{g_0(b-1)} & \theta^{g_1(b-1)} & \cdots &  \theta^{g_{b-1}(b-1)}\\
		\end{array}
		\right)} {\left(\begin{array}{l}  c_{i_pu+g_0,a} \\
			c_{i_pu+g_1,a} \\ \ \ \ \ \ \vdots \\
			c_{i_pu+g_{b-1},a}
		\end{array}
		\right)}={\left(\begin{array}{llll}
			\Delta_{0,a}-\sum\limits_{g\in{\cal J}}c_{i_pu+g,a} \\
			\Delta_{1,a}-\sum\limits_{g\in{\cal J}}\theta^{g}c_{i_pu+g,a}
			\\ \ \ \indent \indent  \ \vdots\\ \Delta_{b-1,a}-\sum\limits_{g\in{\cal J}}\theta^{g(b-1)}c_{i_pu+g,a}, \\
		\end{array}
		\right)},
\end{equation}

\noindent where $a\in[l]$. Since $\theta^{g_i}\neq
\theta^{g_{i'}}$ for $g_i,g_{i'}\in[u]$ with $g_i\neq g_{i'}$, then the
Vandermonde matrix on the left side of the linear system (\ref{xx1})
is invertible,  we can recover $\{c_{i_pu+g_0,a},
c_{i_pu+g_1,a}, \ldots, $ $c_{i_pu+g_{b-1},a}:
a\in[l]\}$, i.e., the $b$ failed nodes in rack
$i_p$. 

In the following, we will discuss the optimality of its
repair bandwidth in two cases depending on the value of $b$.

Case 1: $b\in[1,u-v]$. We show that the code achieve optimal repair
bandwidth for this case. Note that $|{\mathcal V}_i|=|{\mathcal V}_0|$ for $1\leq i\leq
\bar{h}-\delta+1$, and ${\mathcal V}_0,{\mathcal V}_1,\ldots,{\mathcal V}_{\bar{h}-\delta+1}$ forms a
partition of $[0,2^{\bar{h}-\delta+1})$ from Lemma \ref{lm4-1}, which means
\begin{center}
	$\displaystyle|\{a:a\in[l],a|_{\mathcal{S}_p}\in
	{\mathcal V}_0\}|=\frac{l}{\bar{h}-\delta+2}$,
	$\displaystyle\forall p\in[0,\frac{\bar{h}}{\bar{h}-\delta+1})$,
\end{center}
\noindent where $\mathcal{S}_p$ is defined in \eqref{lm4-0}.  In the download phase, host rack $i_p$ ($p\in[0,\bar{h})$) downloads
$\{H_{i_p,j}(a,m):a\in[l],a|_{\mathcal{S}_{\lfloor
		p/(\bar{h}-\delta+1)\rfloor}}$ $\in {\mathcal V}_0, m\in[b]\}$ from every helper rack $j\in {\cal R}$, i.e.,
$bl/(2+\bar{h}-\delta)$ symbols from each of the $\bar{d}$ helper
racks. Consequently, the repair bandwidth at this phase is
\begin{align}\label{4-e4}
	\displaystyle\gamma_{31}=\frac{\bar{d}hl}{2+\bar{h}-\delta}.
\end{align}

Next, during the cooperative phase, each host rack
$i_p$ downloads $H_{i,i_p}(a,m)$, $a\in[l]$ with
$a|_{\mathcal{S}_{\lfloor p/(\bar{h}-\delta+1)\rfloor}}\in {\mathcal V}_0$, $m\in[b]$
from every host rack $i\in \mathcal{S}_{\lfloor p/(\bar{h}-\delta+1)\rfloor}\setminus\{i_p\}$, i.e., $bl/(2+\bar{h}-\delta)$ symbols from each of other
$\bar{h}-\delta$ host racks. Herein, we would like to emphasize that these data
are accessible from Lemma \ref{lm4-3}. So,
the repair bandwidth at this phase is
\begin{align}\label{4-e5}
	\displaystyle\gamma_{32}=\bar{h}(\bar{h}-\delta)\cdot\frac{bl}{2+\bar{h}-\delta}=
	\frac{h(\bar{h}-\delta)l}{2+\bar{h}-\delta}.
\end{align}
Therefore, by \eqref{4-e4} and \eqref{4-e5}, the total repair bandwidth across the two repair phase is
\begin{center}
	$\displaystyle\gamma_3=\gamma_{31}+\gamma_{32}=\frac{h(\bar{d}+\bar{h}-\delta)l}{2+\bar{h}-\delta}$,
\end{center}
\noindent which achieve minimum repair bandwidth from \eqref{e2}.

Case 2: $b\in[u-v+1,u]$. We show that the repair bandwidth of the code $\mathcal{C}$ is asymptotically optimal when $\bar{d}$ is large
enough. During the download phase, for each integer $m\in[0,u-v)$,  by conclusion 1) of Lemma \ref{lm4-2}
each host rack $i_p$ $(p\in[0,\bar{h}))$ downloads
$\{H_{i_p,j}(a,m):a\in[l],a|_{\mathcal{S}_{\lfloor
		p/(\bar{h}-\delta+1)\rfloor}}$ $\in {\mathcal V}_0\}$ from every helper rack
$j\in {\cal R}$, i.e., $(u-v)l/(2+\bar{h}-\delta)$ symbols from each of
the $\bar{d}$ helper racks. For each integer $m\in[u-v,b)$, by conclusion 2) of Lemma \ref{lm4-2}  host
rack $i_p$ downloads
$\{H_{i_p,j}(a,m):a\in[l],a|_{\mathcal{S}_{\lfloor
		p/(\bar{h}-\delta+1)\rfloor}}$ $\in {\mathcal V}_0\}$ from every the helper rack $j\in \mathcal{R}\cup\{j'\}$ where $j'\in[\bar{n}]\setminus(\mathcal{R}\cup\mathcal{F})$, i.e., $(b-u+v)l/(2+\bar{h}-\delta)$
symbols from each of the $\bar{d}+1$ helper racks. So, the
repair bandwidth at this phase is
\begin{align}\label{4-e6}
\nonumber	\displaystyle\gamma_{41}&=\bar{d}\cdot\bar{h}\cdot\frac{(u-v)l}{2+\bar{h}-\delta}+(\bar{d}+1)\cdot\bar{h}\cdot\frac{(b-u+v)l}{2+\bar{h}-\delta}\\
	&=\displaystyle\frac{\bar{d}hl}{2+\bar{h}-\delta}+\frac{\bar{h}(b-u+v)l}{2+\bar{h}-\delta}.
\end{align}
Exactly the same as discussed for case 1, the repair bandwidth at
the cooperative repair phase is
\begin{align}\label{4-e7}
	\displaystyle\gamma_{42}=\bar{h}(\bar{h}-\delta)\cdot\frac{bl}{2+\bar{h}-\delta}=
	\frac{h(\bar{h}-\delta)l}{2+\bar{h}-\delta}.
\end{align}
Therefore, by \eqref{4-e6} and \eqref{4-e7}, the total repair bandwidth across the two repair phase is
\begin{center}
	$\displaystyle\gamma_4=\gamma_{41}+\gamma_{42}=\frac{h(\bar{d}+\bar{h}-\delta)l}{2+\bar{h}-\delta}
	+\frac{\bar{h}(b-u+v)l}{2+\bar{h}-\delta}$.
\end{center}
\noindent Note that $v<u$, we have $b-u+v<b$, then
\begin{center}
	$\displaystyle\gamma_{4}=\frac{h(\bar{d}+\bar{h}-\delta)l}{2+\bar{h}-\delta}
	+\frac{\bar{h}(b-u+v)l}{2+\bar{h}-\delta}<\frac{h(\bar{d}+1+\bar{h}-\delta)l}{2+\bar{h}-\delta}$.
\end{center}
In this case, the ratio of the amount of downloaded symbols to the
optimal repair bandwidth given in \eqref{e2} is smaller than
$1+1/(\bar{d}+\bar{h}-\delta)$. Consequently, when $b>u-v$, if the number of helper racks $\bar{d}$ is sufficiently large
, the repair bandwidth of the code approaches optimal level. \qed

\subsection{Example 2}
\label{subsect-scheme-2} 

In this subsection, we present a concrete example of rack-aware MSCR codes to demonstrate our construction method.  Now let us take the  $(n=45,k=19,l= 2^{15})$ rack-aware MSPCR code to illustrate the method as follows. 

\subsubsection{ Construct an MDS array code $\mathcal{C}$ with $l=2^{15}$} Let $\mathbb{F}$ be a finite field of size 97. Let
$(n,u,k,r,d,h,b,v,\delta)=(45,3,19,26,22,12,$ $2,1,4)$, then we have
$\bar{n}=n/u=15$, $\bar{k}=(k-v)/u=6$, $\bar{r}=(r+v)/u=9$,
$\bar{d}=(d-u+b)/u=7$, $\bar{s}=\bar{d}-\bar{k}+1=2$, $\bar{h}=h/b=6$ and
$l=\bar{s}^{\bar{n}}=2^{15}$. Let $\xi$ be a primitive element of
$\mathbb{F}$ and $\theta$ be an element in $\mathbb{F}$ with multiplicative order
$u=3$. Consider an $(n=45,k=19,l=2^{15},d=22,h=12,\delta=4)$ code $\mathcal{C}=({\bf c}_{0},{\bf c}_{1},\ldots,{\bf c}_{44})$ defined by the following parity check equations over $\mathbb{F}$:
\begin{equation}\label{e13-2}
	\sum\limits_{i=0}^{14}\sum\limits_{g=0}^{2}\theta^{gt}\lambda_{i,a_i}^tc_{3i+g,a}=0,
\end{equation}
\noindent where ${\bf c}_{3i+g}=(c_{3i+g,0},c_{3i+g,1},\ldots,c_{3i+g,2^{15}-1})^\top$ is a column vector  over $\mathbb{F}$ for $i\in[45]$ and $g\in[3]$,	
$t\in[26]$, $a\in[2^{15}]$,
$\lambda_{i,a_i}=\xi^{2i+a_i}$ for $a_i=0,1$.  Similar to the discussion in the first rack-aware MSPCR codes, we can easily know that code $\mathcal{C}$ has the MDS property.  Since $t\in[26]$, taking $t=3\bar{t}+m$, we have $\bar{t}\in[9]$ when $m\in[2]$. Substituting $t=3\bar{t}+m$ into equation \eqref{e13-2}, we can get 
\begin{align}\label{e13-4}
	\sum\limits_{i=0}^{14}\lambda_{i,a_i}^{3\bar{t}+m}\sum\limits_{g=0}^{2}\theta^{gm}c_{3i+g,a}=0,
\end{align}
\noindent where $a\in[2^{15}]$, $\bar{t}\in[9]$, $m\in[2]$.

\subsubsection{ Construct a proper Hamming code ${\mathcal V}_0$} Let the index set of the
$\bar{h}$ host racks be ${\cal F}=\{0,1,\ldots,5\}$, and the index set of the $\bar{d}$ helper racks be
${\cal R}=\{6,7,\ldots,12\}$, the index set of the unconnected racks be ${\cal U}=\{13,14\}$,
and the index set of the $h$ failed nodes be
${\cal F}'=\{3i+g:i\in{\cal F}, g=0,1\}$.  Partition the host racks ${\cal F}$ into two
sets $\mathcal{S}_0=\{0, 1, 2\}$ and $\mathcal{S}_1=\{3,4,5\}$. For
$\bar{h}-\delta+\bar{s}=2^2$ , the two codewords of
$(2^2-1=3,1)$ linear binary Hamming code are ${\mathcal V}_0=\{(0,0,0),(1,1,1)\}=\{0,7\}$. Then, we
have ${\mathcal V}_1={\mathcal V}_0+(1,0,0)=\{(1,0,0),(0,1,1)\}=\{3,4\}$. Similarly, ${\mathcal V}_2=\{(0,1,0),(1,0,1)\}=\{2,5\}$ and ${\mathcal V}_3=\{(0,0,1),(1,1,0)\}=\{1,6\}$. Clearly, $\cup_{i\in[4]}{\mathcal V}_i=[8]$.

\subsubsection{ Each host rack downloads some symbols from each of the helper racks}
For host rack $p\in \mathcal{S}_0$, we can obtain from $(\ref{e13-4})$ that
\begin{align*}
		\lambda_{p,a_p}^{3\bar{t}+m}\sum\limits_{g=0}^{2}\theta^{gm}c_{3p+g,a} +\sum\limits_{i\in{\cal F}\setminus\{p\}}\lambda_{i,a_i}^{3\bar{t}+m}\sum\limits_{g=0}^{2}\theta^{gm}c_{3i+g,a}+\sum\limits_{j\in{\cal R}}\lambda_{j,a_j}^{3\bar{t}+m}\sum\limits_{g=0}^{2}\theta^{gm}c_{3j+g,a} +\sum\limits_{z\in{\cal U}}\lambda_{z,a_z}^{3\bar{t}+m}\sum\limits_{g=0}^{2}\theta^{gm}c_{3z+g,a}=0,
\end{align*}
\noindent and
\begin{align*}
	 \lambda_{p,a_p\oplus1}^{3\bar{t}+m}\sum\limits_{g=0}^{2}\theta^{gm}c_{3p+g,a(p,a_p\oplus1)}+\sum\limits_{i\in{\cal F}\setminus\{p\}}\lambda_{i,a_i}^{3\bar{t}+m}\sum\limits_{g=0}^{2}\theta^{gm}c_{3i+g,a(p,a_p\oplus1)}
	\\+\sum\limits_{j\in{\cal R}}\lambda_{j,a_j}^{3\bar{t}+m}\sum\limits_{g=0}^{2}\theta^{gm}c_{3j+g,a(p,a_p\oplus1)}+\sum\limits_{z\in{\cal U}}\lambda_{z,a_z}^{3\bar{t}+m}\sum\limits_{g=0}^{2}\theta^{gm}c_{3z+g,a(p,a_p\oplus1)}=0,
\end{align*}
\noindent which gives 
\begin{align}\label{e14-1}
	\nonumber	&\lambda_{p,a_p}^{3\bar{t}+m}\sum\limits_{g=0}^{2}\theta^{gm}c_{3p+g,a}+\lambda_{p,a_p\oplus1}^{3\bar{t}+m}\sum\limits_{g=0}^{2}\theta^{gm}c_{3p+g,a(p,a_p\oplus1)}
	 +\sum_{i\in{\cal F}\setminus\{p\}}\lambda_{i,a_i}^{3\bar{t}+m}(\sum_{g=0}^{2}\theta^{gm}c_{3i+g,a}+\sum_{g=0}^{2}\theta^{gm}c_{3i+g,a(p,a_p\oplus1)}) \\
		+&\sum_{z\in{\cal U}}\lambda_{z,a_z}^{3\bar{t}+m}(\sum_{g=0}^{2}\theta^{gm}c_{3z+g,a}+\sum_{g=0}^{2}\theta^{gm}c_{3z+g,a(p,a_p\oplus1)}) 
	=-\sum_{j\in{\cal R}}\lambda_{j,a_j}^{3\bar{t}+m}(\sum_{g=0}^{2}\theta^{gm}c_{3j+g,a}+\sum_{g=0}^{2}\theta^{gm}c_{3j+g,a(p,a_p\oplus1)}),
\end{align}
\noindent where $\bar{t}\in[9]$, $m\in[2]$, $a\in[2^{15}]$ with $(a_2,a_1,a_0)\in {\mathcal V}_0$.

Define 
\begin{align*}
	H_{p,i}(a,m)=\sum_{g=0}^{2}\theta^{gm}c_{3i+g,a}+\sum_{g=0}^{2}\theta^{gm}c_{3i+g,a(p,a_p\oplus1)} 
\end{align*}
\noindent for  $m\in[2]$, and $a\in[2^{15}]$ with $(a_2,a_1,a_0)\in {\mathcal V}_0$. Let
\begin{align*}
	&{\mathbf v}_{p,a,m}^\top\\=&[\sum\limits_{g=0}^{2}\theta^{gm}c_{3p+g,a},\sum\limits_{g=0}^{2}\theta^{gm}c_{3p+g,a(p,a_p\oplus1)},H_{p,0}(a,m), \ldots,H_{p,p-1}(a,m),H_{p,p+1}(a,m),\ldots,H_{p,5}(a,m), H_{p,13}(a,m),H_{p,14}(a,m)],
\end{align*}
\noindent and
\begin{align*}
	{\bf u}_{p,a,m}^\top=[-H_{p,6}(a,m),-H_{p,7}(a,m),\ldots,-H_{p,12}(a,m)],
\end{align*}
\noindent then equation \eqref{e14-1} can be rewritten as
\begin{align*}
	{\bf M}_{p,a,m}{\bf v}_{p,a,m}={\bf M}_{p,a,m}'{\bf u}_{p,a,m},
\end{align*}
\noindent where 
	\[{\bf M}_{p,a,m}=
	\left(
	\begin{array}{cccccccccc}
		\lambda_{p,a_p}^{m}  & \lambda_{p,a_p\oplus1}^{m} & \lambda_{0,a_0}^{m} & \ldots & \lambda_{p-1,a_{p-1}}^{m}	& \lambda_{p+1,a_{p+1}}^{m} & \ldots &
		\lambda_{5,a_5}^{m} & \lambda_{13,a_{13}}^{m}  & \lambda_{14,a_{14}}^{m}\\
		\lambda_{p,a_p}^{3+m}  & \lambda_{p,a_p\oplus1}^{3+m} & \lambda_{0,a_0}^{3+m} & \ldots & \lambda_{p-1,a_{p-1}}^{3+m}	& \lambda_{p+1,a_{p+1}}^{3+m} & \ldots &
		\lambda_{5,a_5}^{3+m} & \lambda_{13,a_{13}}^{3+m}  & \lambda_{14,a_{14}}^{3+m}\\
		\vdots  & \vdots & \vdots & \cdots & \vdots	& \vdots & \cdots &
		\vdots & \vdots & \vdots\\
		\lambda_{p,a_p}^{24+m}  & \lambda_{p,a_p\oplus1}^{24+m} & \lambda_{0,a_0}^{24+m} & \ldots & \lambda_{p-1,a_{p-1}}^{24+m}	& \lambda_{p+1,a_{p+1}}^{24+m} & \ldots &
		\lambda_{5,a_5}^{24+m} & \lambda_{13,a_{13}}^{24+m}  & \lambda_{14,a_{14}}^{24+m}\\
	\end{array}
	\right)
	\]
and
\[ {\bf M}'_{p,a,m}=
\left(
\begin{array}{ccccccc}
	\lambda_{6,a_{6}}^{m} & \lambda_{7,a_{7}}^{m} & \cdots &
	\lambda_{12,a_{12}}^{m}\\
	\lambda_{6,a_{6}}^{3+m} & \lambda_{7,a_{7}}^{3+m} & \cdots &
	\lambda_{12,a_{12}}^{3+m}\\
	\vdots & \vdots & \ddots &
	\vdots\\
	\lambda_{6,a_{6}}^{24+m} & \lambda_{7,a_{7}}^{24+m} & \cdots &
	\lambda_{12,a_{12}}^{24+m}\\
\end{array}
\right).
\]
Note that for any $m\in[2]$ and $a\in[2^{15}]$ with $(a_2,a_1,a_0)\in {\mathcal V}_0$, host rack $p$ downloads	$H_{p,j}(a,m)$
 from all helper racks $j\in{\cal R}$, thus $\mathbf{u}_{p,a,m}$ is
obtained. The matrix ${\bf M}_{p,a,m}$ can be rewritten as
${\bf M}_1$diag$(\lambda_{p,a_{p}}^{m},\lambda_{p,a_p\oplus1}^{m},\lambda_{0,a_{0}}^{m},$ $\ldots,\lambda_{p-1,a_{p-1}}^{m},\lambda_{p+1,a_{p+1}}^{m},\ldots,$ $\lambda_{5,a_{5}}^{m}, \lambda_{13,a_{13}}^{m},\lambda_{14,a_{14}}^{m})$,
where
	\[{\bf M}_1=
	\left(
	\begin{array}{cccccccccc}
		1  & 1 & 1 & \ldots & 1	& 1 & \ldots &
		1 & 1  & 1\\
		\lambda_{p,a_p}^{3}  & \lambda_{p,a_p\oplus1}^{3} & \lambda_{0,a_0}^{3} & \ldots & \lambda_{p-1,a_{p-1}}^{3}	& \lambda_{p+1,a_{p+1}}^{3} & \ldots &
		\lambda_{5,a_5}^{3} & \lambda_{13,a_{13}}^{3}  & \lambda_{14,a_{14}}^{3}\\
		\vdots  & \vdots & \vdots & \cdots & \vdots	& \vdots & \cdots &
		\vdots & \vdots & \vdots\\
		\lambda_{p,a_p}^{24}  & \lambda_{p,a_p\oplus1}^{24} & \lambda_{0,a_0}^{24} & \ldots & \lambda_{p-1,a_{p-1}}^{24}	& \lambda_{p+1,a_{p+1}}^{24} & \ldots &
		\lambda_{5,a_5}^{24} & \lambda_{13,a_{13}}^{24}  & \lambda_{14,a_{14}}^{24}\\
	\end{array}
	\right)
	\]

Since $\lambda_{i,j}=\xi^{2i+j}$, we have $\lambda_{i,j}^3\neq
\lambda_{i',j'}^3$ for all $i,i'\in[15]$, $j,j'\in[2]$ with
$(i,j)\neq(i',j')$, then Vandermonde matrix ${\bf M}_1$ are invertible
for each $a\in[2^{15}]$ with $(a_2,a_1,a_0)\in {\mathcal V}_0$. It follows that the matrix ${\bf M}_{p,a,m}$ is also invertible, therefore
\begin{align*}
	{\bf v}_{p,a,m}={\bf M}_{p,a,m}^{-1}{\bf M}'_{p,a,m}{\bf u}_{p,a,m}, m\in[2],
	a\in[2^{15}] {\rm with} \ (a_2,a_1,a_0)\in {\mathcal V}_0.
\end{align*}
\noindent In other words,  host rack $p$ can recover the data
$\{\sum\limits_{g=0}^{2}\theta^{gm}c_{3p+g,a},\sum\limits_{g=0}^{2}\theta^{gm}c_{3p+g,a(p,a_p\oplus1)}: m\in[2], a\in[2^{15}],(a_2,a_1,a_0)\in {\mathcal V}_0\}$,
and $\{H_{p,i}(a,m): m\in[2],a\in[2^{15}], (a_2,a_1,a_0)\in {\mathcal V}_0\}$ for $i\in{\cal F}\setminus\{p\}$. 

\subsubsection{ Cooperatively repair the failed nodes} For each $p\in \mathcal{S}_0$, the data
$\{\sum\limits_{g=0}^{2}\theta^{gm}c_{3p+g,a(i,a_i\oplus1)}:a\in[2^{15}],(a_2,a_1,a_0)\in {\mathcal V}_0, m\in[2]\}$ can be recovered by the
known data
\begin{center}
	$\{\sum\limits_{g=0}^{2}\theta^{gm}c_{3p+g,a}
	:a\in[2^{15}],(a_2,a_1,a_0)\in
	{\mathcal V}_0,m\in[2]\}$,
\end{center}
\noindent and the data
\begin{center}
	$\{H_{i,p}(a,m)=\sum\limits_{g=0}^{2}\theta^{gm}c_{3p+g,a}+\sum\limits_{g=0}^{2}\theta^{gm}c_{3p+g,a(i,a_{i}\oplus1)}$:
	$a\in[2^{15}],(a_2,a_1,a_0)\in
	{\mathcal V}_0,m\in[2]\}$
\end{center}
\noindent downloaded from host rack $i\in \mathcal{S}_0\setminus\{p\}$. Thus, for each
$p\in \mathcal{S}_0$, host rack $p$ can obtain the
data
\begin{align*}
	&\{\sum\limits_{g=0}^{2}\theta^{gm}c_{3p+g,a},\sum\limits_{g=0}^{2}\theta^{gm}c_{3p+g,a(p,a_p\oplus1)}
	:a\in[2^{15}],m\in[2],(a_2,a_1,a_0)\in
	{\mathcal V}_0\}\\ &
	\bigcup(\bigcup_{i\in \mathcal{S}_0\setminus\{p\}}\{\sum\limits_{g=0}^{2}\theta^{gm}c_{3p+g,a(i,a_i\oplus1)}:a\in[2^{15}],(a_2,a_1,a_0)\in {\mathcal V}_0,m\in[2]\})\\
	&=\{\sum\limits_{g=0}^{2}\theta^{gm}c_{3p+g,a}
	:a\in[2^{15}], m\in[2]\},
\end{align*}
\noindent since $\cup_{i\in[4]}{\mathcal V}_i=[8]$.

In the following, we recover the 2 failed nodes
in host rack $p\in \mathcal{S}_0$. By the assumption that the index
set of the $2$ failed nodes in rack $p$ is
${\cal I}=[2]$. For any $a\in[2^{15}]$ and $m\in[2]$, from the above we know the data
$\Delta_{m,a}=\sum\limits_{g=0}^{2}\theta^{gm}c_{3p+g,a}$, thus
\begin{align*} 
	c_{3p,a}+\theta^{m}c_{3p+1,a}=\Delta_{m,a}-\theta^{2m}c_{3p+2,a}.
\end{align*}
\noindent Taking $m\in[2]$, we can obtain
\begin{equation} \label{e15-1} 
	{\left(\begin{array}{llll}
			1 & 1 \\
			1 & \theta \\
		\end{array}
		\right)} {\left(\begin{array}{l} 
			c_{3p,a} \\
			c_{3p+1,a} \\ 
		\end{array}
		\right)}={\left(\begin{array}{llll}
			\Delta_{0,a}-c_{3p+2,a} \\
			\Delta_{1,a}-\theta^2c_{3p+2,a}
		\end{array}
		\right)},
\end{equation}    
\noindent where $a\in[2^{15}]$. Since $\theta\neq
1$,  matrix 	${\left(\begin{array}{llll}
		1 & 1 \\
		1 & \theta \\
	\end{array}
	\right)}$ in linear system \eqref{e15-1}
is invertible, then we can recover $\{c_{3p,a},
c_{3p+1,a}: a\in[2^{15}]\}$, i.e., the $2$ failed nodes in rack
$p$. 

For $p\in \mathcal{S}_1$, using the same methods, we can recover $\{c_{3p,a},
c_{3p+1,a}: a\in[2^{15}]\}$ which are stored in nodes $3p$ and $3p+1$, respectively.

In the download phase, we download $({2l}/{4})\times6\times7=14l$  symbols from the 7 helper racks, while we download $({2l}/{4})\times6\times2=6l$ in the cooperative phase. Then, during the repair process we download a total of $20l$ symbols, which achieves the lower bound in \eqref{e2}.  A numerical comparisons with existing schemes is presented in table \ref{tab4}.

\begin{table}[http!]
	\renewcommand{\arraystretch}{1.2}
	\setlength\tabcolsep{3.0pt} 
	\centering
	\caption{ The Existing Schemes And Our Schemes For $(n,u,k,d,h,b,v)=(45,3,15,22,12,2,1)$
		\label{tab4}
	} 
	\begin{tabular}{|c|c|c|c|c|c|c|c|c|}	
		\hline
		&  Sub-packetization  $l$ & \tabincell{c}{Optimal repair\\ bandwidth $\gamma$} &  \tabincell{c}{ Field \\ size $|\mathbb{F}|$} \\  
		\hline
		\cite{GDL} & $7\times2^{15}$ & $144\times 2^6$ & 97\\ 
		\hline
		Our schemes &	 $2^{15}$ & $20\times 2^6$& 97\\ 		
		\hline
	\end{tabular}
\end{table}
It is obvious that the sub-packetization of our scheme is $7$ times smaller than that of \cite{GDL}, the optimal repair bandwidth is reduced by more than a factor of $7$ compared to that of \cite{GDL}, while the fields employed in both schemes are of the same size.

\section{Conclusion}

In this paper, we considered the rake-aware MSPCR codes for repairing multiple node failures and derive the lower bound on repair bandwidth from helper racks. We presented two classes of the rake-aware MSPCR codes, and showed that they achieve optimal repair bandwidth when $b\in[1,u-v]$, and asymptotically optimal repair bandwidth when $b\in[u-v+1,u]$. When $\delta=1$, compared to known rack-aware minimum storage (fully) cooperative regenerating (MSCR) codes, the second class of our codes achieves an $\bar{h}+1$ times reduction in sub-packetization by using grouping repair method.

\begin{appendices}
\section{The proof of Theorem \ref{th1}} 
\label{app-th1}
	Let ${\cal F}$ be the index set of  host racks with size
$\bar{h}$, each of which contains $b=h/\bar{h}$ failed nodes, ${\cal R}$ be the index set of
helper racks with size $\bar{d}$, ${\cal F}_\delta$ be any
$(\delta-1)$-subset of ${\cal F}$, ${\cal R}_\delta$ be any
$(\bar{k}-\bar{h}+\delta-1)$-subset of ${\cal R}$. We consider two special
cases where the data collector connects $k$ nodes.

Case 1: There are $\bar{k}-\bar{h}+\delta-1$ relayer nodes in helper
racks of ${\cal R}_{\delta}$, $\bar{h}-\delta+1$ relayer nodes in host racks of ${\cal F}\setminus{\cal F}_\delta$, and $v$ nodes from some racks other than the helper racks and the host racks connected to the data collector. Therefore,  there are $k$ nodes connected to the data collector, where we recall the fact that a relay node can access the contents of other nodes within the same rack and $k=\bar{k}u+v$. Clearly, the $\bar{k}-\bar{h}+\delta-1$ helper racks in ${\cal R}_\delta$ contribute $(\bar{k}-\bar{h}+\delta-1)ul$ symbols to the data collector, and the $v$ nodes  contribute $vl$ symbols to the data collector.
According to the download phase of the cooperative repair scheme, apart from the  $\bar{k}-\bar{h}+\delta-1$  helper racks in ${\cal R}_{\delta}$ that have already been connected to the data collector, each host rack $i$ in ${\cal F}\setminus{\cal F}_\delta$  downloads  $\beta_1(i,j)$ symbols from each helper rack $j$ in ${\cal R}\setminus{\cal R}_{\delta}$,  where $\beta_1(i,j)$ is the number of symbols downloaded
by rack $i$ from the nodes in rack $j$ in the download phase. Without loss of generating, each host rack of ${\cal F}\setminus{\cal F}_\delta$ exchanges data only with $\bar{h}-\delta$ other host racks in ${\cal F}\setminus{\cal F}_\delta$. We know that each host rack $i$ in ${\cal F}\setminus{\cal F}_\delta$  contributes $\sum\limits_{j\in{\cal R}\setminus{\cal R}_\delta}\beta_1(i,j)+(u-b)l$ symbols to the data collector. From this observation, we can conclude that the  $\bar{h}-\delta+1$ host racks in ${\cal F}\setminus{\cal F}_\delta$  contribute $\sum\limits_{i\in{\cal F}\setminus{\cal F}_\delta}\sum\limits_{j\in{\cal R}\setminus{\cal R}_\delta}\beta_1(i,j)+(\bar{h}-\delta+1)(u-b)l$ symbols to the data collector. Therefore, the number of symbols contributed to the data collector by the connected 
$k$ nodes is
\begin{align*}
	\sum\limits_{i\in{\cal F}\setminus{\cal F}_\delta}\sum\limits_{j\in{\cal R}\setminus{\cal R}_\delta}&\beta_1(i,j) +(\bar{h}-\delta+1)(u-b)l +(\bar{k}-\bar{h}+\delta-1)ul+vl.
\end{align*}
\noindent From the MDS property of the code $\mathcal{C}$, we have
\begin{align*}
	\displaystyle\sum\limits_{i\in{\cal F}\setminus{\cal F}_\delta}\sum\limits_{j\in{\cal R}\setminus
		{\cal R}_\delta}&\beta_1(i,j)+(u-b)(\bar{h}-\delta+1)l +(\bar{k}-\bar{h}+\delta-1)ul+vl\geq |\mathbf{w}|.
\end{align*}
\noindent Note that $|\mathbf{w}|=kl$, we can get

\begin{center} 
	$\displaystyle\sum\limits_{i\in{\cal F}\setminus{\cal F}_\delta}\sum\limits_{j\in{\cal R}\setminus
		{\cal R}_\delta}\beta_1(i,j)\geq (\bar{h}-\delta+1)bl$.
\end{center}
Firstly, by adding up the left-hand side with respect to all $(\delta-1)$-subsets of ${\cal F}$, we obtain	
\begin{align*}
	{{\bar{h}-1}\choose{\delta-1}}\sum\limits_{i\in{\cal F}}\sum\limits_{j\in{\cal R}\setminus
		{\cal R}_\delta}\beta_1(i,j)=&\sum\limits_{{\cal F}_\delta\subset{\cal F},i\not\in{\cal F}_\delta}\sum\limits_{j\in{\cal R}\setminus
		{\cal R}_\delta}\sum\limits_{i\in{\cal F}}\beta_1(i,j)\\
	=&\displaystyle\sum\limits_{{\cal F}_\delta\subset{\cal F}}\sum\limits_{j\in{\cal R}\setminus
		{\cal R}_\delta}\sum\limits_{i\in{\cal F}\setminus{\cal F}_\delta}\beta_1(i,j)
	\geq\displaystyle {\bar{h}\choose{\delta-1}}(\bar{h}-\delta+1)bl.
\end{align*}
\noindent This implies that
\begin{equation}\label{e3}
	\sum\limits_{j\in{\cal R}\setminus
		{\cal R}_\delta}\sum\limits_{i\in{\cal F}}\beta_1(i,j)\geq hl.
\end{equation}

Secondly, by adding up the left-hand side with respect to all $(\bar{k}-\bar{h}+\delta-1)$-subsets of ${\cal R}$, we get
\begin{align*}
	\sum\limits_{{\cal R}_\delta\subset{\cal R}}\sum\limits_{j\in{\cal R}\setminus
		{\cal R}_\delta}\sum\limits_{i\in{\cal F}}\beta_1(i,j)=
	{{\bar{d}-1}\choose{\bar{k}-\bar{h}+\delta-1}}\sum\limits_{j\in{\cal R}}\sum\limits_{i\in{\cal F}}\beta_1(i,j) \geq{{\bar{d}}\choose{\bar{k}-\bar{h}+\delta-1}}hl,
\end{align*}
\noindent that is
\begin{equation}\label{e4}
	\sum\limits_{j\in{\cal R}}\sum\limits_{i\in{\cal F}}\beta_1(i,j)\geq
	\frac{\bar{d}hl}{\bar{d}-\bar{k}+\bar{h}-\delta+1}.
\end{equation}
Next, we aim to show that the equality in (\ref{e4}) holds if and only if
$\sum\limits_{i\in{\cal F}}\beta_1(i,j)=
\frac{hl}{\bar{d}-\bar{k}+\bar{h}-\delta+1}$. If
$\sum\limits_{i\in{\cal F}}\beta_1(i,j)=
\frac{hl}{\bar{d}-\bar{k}+\bar{h}-\delta+1}$ for every $j\in{\cal R}$, 
it is evident that the equality in (\ref{e4}) holds.
In the following, we show that if
the equality in (\ref{e4}) holds, then $\sum\limits_{i\in{\cal F}}\beta_1(i,j)=
\frac{hl}{\bar{d}-\bar{k}+\bar{h}-\delta+1}$ for each $j\in{\cal R}$.
Otherwise, assume that the equality in (\ref{e4}) holds while 
$\sum\limits_{i\in{\cal F}}\beta_1(i,j)<
\frac{hl}{\bar{d}-\bar{k}+\bar{h}-\delta+1}$ for some $j\in{\cal R}$. 
Note that the equality in (\ref{e4}) holds if and only if the equality in
(\ref{e3}) holds, that is, the sum of all symbols downloaded from any
$\bar{d}-\bar{k}+\bar{h}-\delta+1 $ helper racks is $hl$. Let
${\cal J}_1\subset{\cal R}$ with $|{\cal J}_1|=\bar{d}-\bar{k}+\bar{h}-\delta+1$ and
$j\in {\cal J}_1$. So,
$\sum\limits_{j\in{\cal J}_1}\sum\limits_{i\in{\cal F}}\beta_1(i,j)=hl$. Since
$\sum\limits_{i\in{\cal F}}\beta_1(i,j)<
\frac{hl}{\bar{d}-\bar{k}+\bar{h}-\delta+1}$, there exists
$j'\in{\cal J}_1$ such that $\sum\limits_{i\in{\cal F}}\beta_1(i,j')>
\frac{hl}{\bar{d}-\bar{k}+\bar{h}-\delta+1}$. Let ${\cal J}_2$ be another
subset of ${\cal R}$ such that $|{\cal J}_2|=\bar{d}-\bar{k}+\bar{h}-\delta+1$,
$j'\in{\cal J}_2$ and $j\not\in{\cal J}_2$. We also have
$\sum\limits_{j\in{\cal J}_2}\sum\limits_{i\in{\cal F}}\beta_1(i,j)=hl$. Set
${\cal J}_3=({\cal J}_2\setminus\{j'\})\cup\{j\}$. Then we have
$|{\cal J}_3|=\bar{d}-\bar{k}+\bar{h}-\delta+1$ and
$\sum\limits_{j\in{\cal J}_2}\sum\limits_{i\in{\cal F}}\beta_1(i,j)<hl$. This
contradicts to (\ref{e3}).

Case 2:  Let ${\cal R}_1$ be any subset of ${\cal R}$ of size $\bar{k}-1$, and
$\mathcal{S}$ be any subset of ${\cal F}\setminus\{i\}$ of size $\bar{h}-\delta$ for any given host rack $i$.
There are $\bar{k}-1$ relayer nodes in helper racks in ${\cal R}_1$, a relayer node in  host rack $i$, and $v$ nodes from some racks other than
the helper racks and the host racks connected to
the data collector, i.e., the data collector is connected to $k$ nodes.
Clearly, the $\bar{k}-1$ helper racks in ${\cal R}_1$ contribute $(\bar{k}-1)ul$ symbols to the data collector, and the $v$ nodes  contribute $vl$ symbols to the data collector.
According to the download phase of the repair scheme, apart from the  $\bar{k}-1$  helper racks in ${\cal R}_{\delta}$ that have already been connected to the data collector, host rack $i$ downloads  $\beta_1(i,j)$ symbols from each helper rack $j$ in ${\cal R}\setminus{\cal R}_1$. From the cooperative phase of the repair scheme, each host rack $i'$ in $\mathcal{S}$ sends $\beta_2(i,i')$ symbols to  host rack $i$, where $\beta_2(i,i')$ represents the number of symbols that host rack $i'$  transfer to  host rack $i$ during the cooperative phase.
We know that host rack $i$ contributes
$\sum\limits_{j\in{\cal R}\setminus{\cal R}_1}\beta_1(i,j)+\sum\limits_{i'\in\mathcal{S}}\beta_2(i,i')+(u-b)l$ symbols to the data collector. Therefore, the number of symbols contributed to the data collector by the connected 
$k$ nodes is
\begin{center}
	$\displaystyle\sum\limits_{j\in{\cal R}\setminus{\cal R}_1}\beta_1(i,j)+\sum\limits_{i'\in\mathcal{S}}\beta_2(i,i')+(u-b)l+(\bar{k}-1)ul+vl$.
\end{center}
\noindent   By the MDS property of the code $\mathcal{C}$, we have
\begin{center}
	$\displaystyle\sum\limits_{j\in{\cal R}\setminus
		{\cal R}_1}\beta_1(i,j)+\sum\limits_{i'\in\mathcal{S}}\beta_2(i,i')+(u-b)l+(\bar{k}-1)ul+vl\geq |\mathbf{w}|$,
\end{center}
\noindent that is	
\begin{center} $\displaystyle\sum\limits_{j\in{\cal R}\setminus
		{\cal R}_1}\beta_1(i,j)+\sum\limits_{i'\in\mathcal{S}}\beta_2(i,i')\geq bl$.
\end{center}
\noindent By adding up the left-hand side on all $i$ in ${\cal F}$, we
get
\begin{align*} \displaystyle\sum\limits_{i\in{\cal F}}\sum\limits_{j\in{\cal R}\setminus
		{\cal R}_1}\beta_1(i,j)+\sum\limits_{i\in{\cal F}}\sum\limits_{i'\in\mathcal{S}}\beta_2(i,i')\geq
	hl.
\end{align*}
By adding up the left-hand side on all
$(\bar{k}-1)$-subsets of ${\cal R}$, we have
\begin{align*}
	&\sum\limits_{{\cal R}_1\subset{\cal R}}\sum\limits_{j\in{\cal R}\setminus
		{\cal R}_1}\sum\limits_{i\in{\cal F}}\beta_1(i,j)+\sum\limits_{{\cal R}_1\subset{\cal R}}\sum\limits_{i\in{\cal F}}\sum\limits_{i'\in\mathcal{S}}\beta_2(i,i')\\
	=&{\bar{d}-1\choose{\bar{k}-1}}\sum\limits_{j\in{\cal R}}\sum\limits_{i\in{\cal F}}\beta_1(i,j)+ {\bar{d}\choose{\bar{k}-1}}\sum\limits_{i\in{\cal F}}\sum\limits_{i'\in\mathcal{S}}\beta_2(i,i')\\
	\geq &{\bar{d}\choose{\bar{k}-1}}hl,
\end{align*}
\noindent that is
\begin{align*}
	\displaystyle\sum\limits_{j\in{\cal R}}\sum\limits_{i\in{\cal F}}\beta_1(i,j)+
	\frac{\bar{d}}{\bar{d}-\bar{k}+1}\sum\limits_{i\in{\cal F}}\sum\limits_{i'\in\mathcal{S}}\beta_2(i,i')\geq
	\frac{\bar{d}hl}{\bar{d}-\bar{k}+1}.
\end{align*}
\noindent From the fact that the equality in (\ref{e4}) holds if and only
if $\sum\limits_{i\in{\cal F}}\beta_1(i,j)=
\frac{hl}{\bar{d}-\bar{k}+\bar{h}-\delta+1}$, we can get
\begin{align*}
	\displaystyle\frac{\bar{d}}{\bar{d}-\bar{k}+1}\displaystyle{\sum\limits_{i\in{\cal F}}}\sum\limits_{i'\in\mathcal{S}}\beta_2(i,i')\geq
	\frac{\bar{d}hl}{\bar{d}-\bar{k}+1}-\frac{\bar{d}hl}{\bar{d}-\bar{k}+\bar{h}-\delta+1},
\end{align*}that is
\begin{equation}\label{e5}
	\sum\limits_{i\in{\cal F}}\sum\limits_{i'\in\mathcal{S}}\beta_2(i,i')\geq
	\frac{h(\bar{h}-\delta)l}{\bar{d}-\bar{k}+\bar{h}-\delta+1}.
\end{equation} 
Combining (\ref{e4}) with (\ref{e5}), we can obtain
\begin{align}\label{e5-1}
	\sum\limits_{j\in{\cal R}}\sum\limits_{i\in{\cal F}}\beta_1(i,j)+\sum\limits_{i\in{\cal F}}\sum\limits_{i'\in\mathcal{S}}\beta_2(i,i')\geq
	\frac{h(\bar{d}+\bar{h}-\delta)l}{\bar{d}-\bar{k}+\bar{h}-\delta+1}.
\end{align} 

Next, we aim to show that when $\beta_1(i,j)=\beta_1$ and $\beta_2(i,i')=\beta_2$ for any $j\in\mathcal{R}$, $i,i'\in\mathcal{F}$ with $i\not=i'$,  the equality in \eqref{e5-1} holds if and only if $\beta_1=\beta_2=
\frac{bl}{\bar{d}-\bar{k}+\bar{h}-\delta+1}$. If $\beta_1=\beta_2=
\frac{bl}{\bar{d}-\bar{k}+\bar{h}-\delta+1}$. it is obvious that the equality in \eqref{e5-1} holds. Next, we show that if
the equality in \eqref{e5-1} holds, then $\beta_1=\beta_2=
\frac{bl}{\bar{d}-\bar{k}+\bar{h}-\delta+1}$. From \eqref{e4} and \eqref{e5}, we have 
$\sum_{j\in{\cal R}}\sum_{i\in{\cal F}}\beta_1(i,j)=\frac{\bar{d}hl}{\bar{d}-\bar{k}+\bar{h}-\delta+1}$ and 	$\sum_{i'\in\mathcal{S}}\sum_{i\in{\cal F}}\beta_2(i,i')=\frac{(\bar{h}-\delta)hl}{\bar{d}-\bar{k}+\bar{h}-\delta+1}$ since the eqality in \eqref{e5-1} holds. Since $\beta_1(i,j)=\beta_1$ and $\beta_2(i,i')=\beta_2$ for any $j\in\mathcal{R}$, $i,i'\in\mathcal{F}$ with $i\not=i'$, then $\beta_1=\beta_2=
\frac{bl}{\bar{d}-\bar{k}+\bar{h}-\delta+1}$.	 \qed
\end{appendices}

\bibliographystyle{IEEEtran}
\bibliography{reference}

\begin{thebibliography}{10}
\providecommand{\url}[1]{#1}
\csname url@samestyle\endcsname
\providecommand{\newblock}{\relax}
\providecommand{\bibinfo}[2]{#2}
\providecommand{\BIBentrySTDinterwordspacing}{\spaceskip=0pt\relax}
\providecommand{\BIBentryALTinterwordstretchfactor}{4}
\providecommand{\BIBentryALTinterwordspacing}{\spaceskip=\fontdimen2\font plus
\BIBentryALTinterwordstretchfactor\fontdimen3\font minus
  \fontdimen4\font\relax}
\providecommand{\BIBforeignlanguage}[2]{{%
\expandafter\ifx\csname l@#1\endcsname\relax
\typeout{** WARNING: IEEEtran.bst: No hyphenation pattern has been}%
\typeout{** loaded for the language `#1'. Using the pattern for}%
\typeout{** the default language instead.}%
\else
\language=\csname l@#1\endcsname
\fi
#2}}
\providecommand{\BIBdecl}{\relax}
\BIBdecl

\bibitem{DGWWR}
A.~G. Dimakis, P.~B. Godfrey, Y.~Wu, M.~J. Wainwright, and K.~Ramchandran,
  ``Network coding for distributed storage systems,'' \emph{IEEE Transactions
  on Information Theory}, vol.~56, no.~9, pp. 4539--4551, 2010.

\bibitem{Rashmi2011}
K.~V. Rashmi, N.~B. Shah, and P.~V. Kumar, ``Optimal exact-regenerating codes
  for distributed storage at the msr and mbr points via a product-matrix
  construction,'' \emph{IEEE Transactions on Information Theory}, vol.~57,
  no.~8, pp. 5227--5239, 2011.

\bibitem{Ernvall2014}
T.~Ernvall, ``Codes between mbr and msr points with exact repair property,''
  \emph{IEEE Transactions on Information Theory}, vol.~60, no.~11, pp.
  6993--7005, Nov. 2014.

\bibitem{Mahdaviani2019}
K.~Mahdaviani, A.~Khisti, and S.~Mohajer, ``Bandwidth adaptive \& error
  resilient mbr exact repair regenerating codes,'' \emph{IEEE Transactions on
  Information Theory}, vol.~65, no.~5, pp. 2736--2759, May 2019.

\bibitem{Zhang2020}
X.~Zhang and Y.~Hu, ``Efficient storage scaling for mbr and msr codes,''
  \emph{IEEE Access}, vol.~8, pp. 78\,992--79\,002, 2020.

\bibitem{Wang2023}
X.~Wang and Y.~Liao, ``New storage codes between the msr and mbr points through
  block designs,'' \emph{IEEE Access}, vol.~11, pp. 87\,120--87\,130, 2023.

\bibitem{GFV}
S.~Goparaju, A.~Fazeli, and A.~Vardy, ``Minimum storage regenerating codes for
  all parameters,'' \emph{IEEE Transactions on Information Theory}, vol.~63,
  no.~10, pp. 6318--6328, 2017.

\bibitem{LLT}
J.~Li, Y.~Liu, and X.~Tang, ``A systematic construction of mds codes with small
  sub-packetization level and near-optimal repair bandwidth,'' \emph{IEEE
  Transactions on Information Theory}, vol.~67, no.~4, pp. 2162--2180, 2021.

\bibitem{LTT}
J.~Li, X.~Tang, and C.~Tian, ``A generic transformation to enable optimal
  repair in mds codes for distributed storage systems,'' \emph{IEEE
  Transactions on Information Theory}, vol.~64, no.~9, pp. 6257--6267, 2018.

\bibitem{LWHY}
G.~Li, N.~Wang, S.~Hu, and M.~Ye, ``Msr codes with linear field size and
  smallest sub-packetization for any number of helper nodes,'' \emph{IEEE
  Transactions on Information Theory}, vol.~70, no.~11, pp. 7790--7806, 2024.

\bibitem{RSE}
N.~Raviv, N.~Silberstein, and T.~Etzion, ``Constructions of high-rate minimum
  storage regenerating codes over small fields,'' \emph{IEEE Transactions on
  Information Theory}, vol.~63, no.~4, pp. 2015--2038, 2017.

\bibitem{TYLH}
X.~Tang, B.~Yang, J.~Li, and H.~D.~L. Hollmann, ``A new repair strategy for the
  hadamard minimum storage regenerating codes for distributed storage
  systems,'' \emph{IEEE Transactions on Information Theory}, vol.~61, no.~10,
  pp. 5271--5279, 2015.

\bibitem{TWB}
I.~Tamo, Z.~Wang, and J.~Bruck, ``Zigzag codes: Mds array codes with optimal
  rebuilding,'' \emph{IEEE Transactions on Information Theory}, vol.~59, no.~3,
  pp. 1597--1616, 2013.

\bibitem{VBV}
M.~Vajha, S.~B. Balaji, and P.~Vijay~Kumar, ``Small-d msr codes with optimal
  access, optimal sub-packetization, and linear field size,'' \emph{IEEE
  Transactions on Information Theory}, vol.~69, no.~7, pp. 4303--4332, 2023.

\bibitem{YB1}
M.~Ye and A.~Barg, ``Explicit constructions of high-rate mds array codes with
  optimal repair bandwidth,'' \emph{IEEE Transactions on Information Theory},
  vol.~63, no.~4, pp. 2001--2014, 2017.

\bibitem{YB2}
------, ``Explicit constructions of optimal-access mds codes with nearly
  optimal sub-packetization,'' \emph{IEEE Transactions on Information Theory},
  vol.~63, no.~10, pp. 6307--6317, 2017.

\bibitem{HLSH}
H.~Hou, P.~P.~C. Lee, K.~W. Shum, and Y.~Hu, ``Rack-aware regenerating codes
  for data centers,'' \emph{IEEE Transactions on Information Theory}, vol.~65,
  no.~8, pp. 4730--4745, 2019.

\bibitem{C-Barg}
Z.~Chen and A.~Barg, ``Explicit constructions of msr codes for clustered
  distributed storage: The rack-aware storage model,'' \emph{IEEE Transactions
  on Information Theory}, vol.~66, no.~2, pp. 886--899, 2020.

\bibitem{HLH}
H.~Hou, P.~P.~C. Lee, and Y.~S. Han, ``Minimum storage rack-aware regenerating
  codes with exact repair and small sub-packetization,'' in \emph{2020 IEEE
  International Symposium on Information Theory (ISIT)}, 2020, pp. 554--559.

\bibitem{ZZ}
L.~Zhou and Z.~Zhang, ``Explicit construction of minimum storage rack-aware
  regenerating codes for all parameters,'' in \emph{2020 IEEE Information
  Theory Workshop (ITW)}, 2021, pp. 1--5.

\bibitem{WC}
J.~Wang and Z.~Chen, ``Rack-aware minimum-storage regenerating codes with
  optimal access,'' \emph{arXiv:2304.08747}, 2023.

\bibitem{Cadambe2013}
V.~R. Cadambe, S.~A. Jafar, H.~Maleki, K.~Ramchandran, and C.~Suh, ``Asymptotic
  interference alignment for optimal repair of mds codes in distributed
  storage,'' \emph{IEEE Transactions on Information Theory}, vol.~59, no.~5,
  pp. 2974--2987, May 2013.

\bibitem{HXWZL}
Y.~Hu, Y.~Xu, X.~Wang, C.~Zhan, and P.~Li, ``Cooperative recovery of
  distributed storage systems from multiple losses with network coding,''
  \emph{IEEE Journal on Selected Areas in Communications}, vol.~28, no.~2, pp.
  268--276, 2010.

\bibitem{SH}
K.~W. Shum and Y.~Hu, ``Cooperative regenerating codes,'' \emph{IEEE
  Transactions on Information Theory}, vol.~59, no.~11, pp. 7229--7258, 2013.

\bibitem{Li2014}
J.~Li and B.~Li, ``Cooperative repair with minimum-storage regenerating codes
  for distributed storage,'' in \emph{IEEE INFOCOM 2014 - IEEE Conference on
  Computer Communications}.\hskip 1em plus 0.5em minus 0.4em\relax IEEE, Apr.
  2014.

\bibitem{LCT}
Y.~Liu, H.~Cai, and X.~Tang, ``A new cooperative repair scheme with k + 1
  helper nodes for (n, k) hadamard msr codes with small sub-packetization,''
  \emph{IEEE Transactions on Information Theory}, vol.~69, no.~5, pp.
  2820--2829, 2023.

\bibitem{YB}
M.~Ye and A.~Barg, ``Cooperative repair: Constructions of optimal mds codes for
  all admissible parameters,'' \emph{IEEE Transactions on Information Theory},
  vol.~65, no.~3, pp. 1639--1656, 2019.

\bibitem{Ye}
M.~Ye, ``New constructions of cooperative msr codes: Reducing node size to
  exp(o(n)),'' \emph{IEEE Transactions on Information Theory}, vol.~66, no.~12,
  pp. 7457--7464, 2020.

\bibitem{ZZW}
Y.~Zhang, Z.~Zhang, and L.~Wang, ``Explicit constructions of optimal-access
  mscr codes for all parameters,'' \emph{IEEE Communications Letters}, vol.~24,
  no.~5, pp. 941--945, 2020.

\bibitem{Zhang2025}
Z.~Zhang, G.~Li, and S.~Hu, ``Constructing (h, d) cooperative msr codes with
  sub-packetization $(d-k+h)(d-k+1)^{n/2}$,'' \emph{IEEE Transactions on
  Information Theory}, vol.~71, no.~4, pp. 2505--2516, 2025.

\bibitem{LSL}
S.~Liu, K.~W. Shum, and C.~Li, ``Exact-repair codes with partial collaboration
  in distributed storage systems,'' \emph{IEEE Transactions on Communications},
  vol.~68, no.~7, pp. 4012--4021, 2020.

\bibitem{LO}
S.~Liu and F.~Oggier, ``On storage codes allowing partially collaborative
  repairs,'' in \emph{2014 IEEE International Symposium on Information Theory},
  2014, pp. 2440--2444.

\bibitem{ZLC}
Y.~Zhang, S.~Liu, and L.~Chen, ``Maximum distance separable array codes
  allowing partial collaboration,'' \emph{IEEE Communications Letters},
  vol.~24, no.~8, pp. 1612--1615, 2020.

\bibitem{LWCT}
Y.~Liu, Y.~Wang, H.~Cai, and X.~Tang, ``Minimum storage partially cooperative
  regenerating codes with small sub-packetization,'' \emph{IEEE Transactions on
  Communications}, vol.~72, no.~1, pp. 38--49, 2024.

\bibitem{VFF}
A.~Vahdat, M.~Al-Fares, N.~Farrington, R.~N. Mysore, G.~Porter, and
  S.~Radhakrishnan, ``Scale-out networking in the data center,'' \emph{IEEE
  Micro}, vol.~30, no.~4, pp. 29--41, 2010.

\bibitem{Zhou2022}
L.~Zhou and Z.~Zhang, ``Rack-aware regenerating codes with multiple erasure
  tolerance,'' \emph{IEEE Transactions on Communications}, vol.~70, no.~7, pp.
  4316--4326, 2022.

\bibitem{WZLT}
J.~Wang, D.~Zheng, S.~Li, and X.~Tang, ``Rack-aware msr codes with error
  correction capability for multiple erasure tolerance,'' \emph{IEEE
  Transactions on Information Theory}, vol.~69, no.~10, pp. 6428--6442, 2023.

\bibitem{GDL}
S.~Gupta, B.~R. Devi, and V.~Lalitha, ``On rack-aware cooperative regenerating
  codes and epsilon-mscr codes,'' \emph{IEEE Journal on Selected Areas in
  Information Theory}, vol.~3, no.~2, pp. 362--378, 2022.

\end{thebibliography}

	\newpage
\end{document}